\documentclass[aps,12pt, amsfonts, amsmath, amssymb, nofootinbib]{revtex4}
\usepackage{amsthm}
\usepackage{multirow}
\usepackage{color}		
\usepackage{graphicx}
\usepackage{esvect}
\usepackage{accents}
\usepackage{tikz}
\begin{document}
\newcommand{\pic}{$\spadesuit$}
\newcommand{\N}{\mathbb{N}}
\newcommand{\Z}{\mathbb{Z}}
\newcommand{\C}{\mathbb{C}}
\newcommand{\R}{\mathbb{R}}
\newcommand{\F}{\mathbb{F}}
\newcommand{\defineL}{\,\mathrel{\mathop:}=}
\newcommand{\defineR}{ =\mathrel{\mathop:}\,}
\newcommand{\be}{\begin{equation}}
\newcommand{\ee}{\end{equation}}
\newcommand{\bi}{\begin{itemize}}
\newcommand{\ei}{\end{itemize}}
\newcommand{\with}{\qquad\text{with}\qquad}
\newcommand{\mand}{\qquad\text{and}\qquad}
\newcommand{\sep}{\; \text{,}\qquad}
\newcommand{\ssep}{\: \text{,}\quad}
\newcommand{\pd}{\partial}
\newcommand{\pdo}{\overline{\partial}}
\newcommand{\ov}[1]{\overline{#1}}
\newcommand{\inv}[1]{\frac{1}{#1}}
\newcommand{\tinv}[1]{\tfrac{1}{#1}}
\newcommand{\abs}[1]{|#1|}
\newcommand{\B}[1]{\mathbf{#1}}
\newcommand{\op}[1]{\text{#1}}
\newcommand{\bracket}[1]{\left( #1 \right)}
\newcommand{\bracketi}[1]{\bigl( #1 \bigr)}
\newcommand{\bracketii}[1]{\Bigl( #1 \Bigr)}
\newcommand{\bracketiii}[1]{\biggl( #1 \biggr)}
\newcommand{\bracketiv}[1]{\Biggl( #1 \Biggr)}
\newcommand{\bpm}{\begin{pmatrix}}
\newcommand{\epm}{\end{pmatrix}}
\newcommand{\bnm}{\begin{matrix}}
\newcommand{\enm}{\end{matrix}}
\newcommand{\w}{\wedge}
\newcommand{\tp}{\otimes}
\newcommand{\ds}{\oplus}
\newcommand{\no}{\nonumber}
\newcommand{\DD}{\mathfrak{D}}
\newcommand{\LL}{\mathfrak{L}}
\newcommand{\FF}{\mathfrak{F}}
\newcommand{\dd}{\mathfrak{d}}
\newcommand{\ff}{\mathfrak{f}}
\newcommand{\cc}{\mathfrak{c}}
\newcommand{\nn}{\mathfrak{n}}
\newcommand{\mm}{\mathfrak{m}}
\newcommand{\coder}{\mathfrak{coder}}
\newcommand{\cohom}{\mathfrak{cohom}}
\newcommand{\morph}{\mathfrak{morph}}
\newcommand{\liftd}[1]{\widehat{#1}}
\newcommand{\lifte}[1]{e^{#1}}
\newcommand{\vect}[1]{\accentset{\rightharpoonup}{#1}}
\newcommand{\T}{\mathsf{T}}

\newtheorem{Def}{Definition}
\newtheorem{Thm}{Theorem}
\newtheorem{Lem}{Lemma}

\title{Homotopy Classification of Bosonic String Field Theory\vspace{1.0cm}}

\author{Korbinian M\"unster}
\email{korbinian.muenster@physik.uni-muenchen.de} 
\author{Ivo Sachs}
\email{ivo.sachs@physik.uni-muenchen.de}
\affiliation{Arnold Sommerfeld Center for Theoretical Physics, Theresienstrasse 37, D-80333 Munich, Germany}

\date{\today}

\begin{abstract}
\vspace{1.7cm}
\centerline{\bf Abstract \vspace{0.4cm}}
We prove the decomposition theorem for the  loop homotopy algebra of quantum closed string field theory and use it to show that closed string field theory is unique up to gauge transformations on a given string background and given S-matrix. For the theory of open and closed strings we use results in open-closed homotopy algebra to show that  the space of inequivalent open string field theories is isomorphic to the space of classical closed string backgrounds. As a further application of the open-closed homotopy algebra we show that string field theory is background independent and locally unique in a very precise sense. Finally we discuss topological string theory in the framework of homotopy algebras and find a generalized correspondence between closed strings and open string field theories.
\end{abstract}
\maketitle
\thispagestyle{empty}
\newpage
\tableofcontents
\newpage


\section{Introduction and Summary}\label{sec:intro}
Historically, the first consistent, interacting formulation of string field theory is Witten's open cubic string field theory \cite{Witten cubic, Leclair cubic,Thorn cubic}. Its algebraic structure is rather simple:  The BRST differential $Q$ and the star product $\ast$, which define the kinetic term and the cubic interaction respectively, satisfy the axioms of a differential graded associative algebra (DGA). More generally, it turns out  \cite{Zwiebach open1} that any formulation of open string field theory realizes an $A_\infty$-algebra, a generalization of a DGA where associativity holds only up to homotopy. 

The general procedure of constructing covariant string field theory, as described by Zwiebach \cite{Zwiebach closed,Zwiebach open-closed}, requires a decomposition of the relevant moduli space of Riemann surfaces into elementary vertices  and graphs. This decomposition guarantees a single cover of moduli space via Feynman rules  and implies that the vertices satisfy a BV master equation. In a second step one employs the operator formalism of the world sheet conformal field theory to construct a morphism of BV algebras from the moduli space to the space of  multilinear functions on the (restricted) state space of the conformal field theory. This is where background dependence enters in the construction. 

At the classical level, the multilinear maps on the state space of the CFT satisfy the axioms of an $A_\infty$- (open string) or $L_\infty$- (closed string) algebra. The classification of physically inequivalent string field theories is then obtained with the help of the decomposition theorem \cite{Kontsevich, Kajiura open1}. This theorem establishes an isomorphism between a given homotopy algebra and the direct sum of a linear contractible algebra and a minimal model. In the context of string field theory, the structure maps of the minimal model are identical to the tree-level $S$-matrix elements of the perturbative string theory in the string background corresponding to the trivial Maurer-Cartan element of the homotopy algebra  \cite{Kajiura open1, Kajiura open2}. 

One purpose of this paper is to extend this classification to quantum closed SFT. To this end we proof the decomposition theorem for loop homotopy algebras, which are a special case of $IBL_\infty$-algebras. We then utilize the decomposition theorem to show that string field theory is unique up to gauge transformations on a given string background. More precisely, two string field theories constructed on the same string background, in particular inducing the same S-matrix, are connected by a 1-parameter family of strong $IBL_\infty$-isomorphisms. This is the algebraic counterpart of the statement that the string vertices at the geometric level define an unique element in the cohomology of the boundary operator plus BV operator on the moduli space of Riemann surfaces \cite{ZS background1, ZS background2, Costello}. 

Given the above result one is naturally led to ask if changes in the closed string background are the only non-trivial deformations of closed string field theory compatible with the operator formalism. We will answer this question within the restriction to deformations which leave the state space of the CFT invariant. In this case we will first establish background independence which amounts to proving that 
shifts in the closed string background are equivalent to conjugation by Maurer-Cartan elements of the homotopy algebra. Since such transformations correspond to weak  $IBL_\infty$-isomorphisms we can define bigger equivalence classes where different closed string backgrounds are identified. We then establish uniqueness of closed string field theory in the sense that there is no non-trivial infinitesimal deformation of closed string field theory compatible with the operator formalism.  

Next we turn to open-closed string field theory. The reformulation of open-closed SFT in terms of homotopy algebras (see \cite{Kajiura open-closed1, Kajiura open-closed2} for the classical case and \cite{qocha} for the quantum theory) relates (quantum) closed, open and open-closed vertices of the SFT to structure maps of $(IB)L_\infty$-,  $A_\infty$-algebras and  $(IB)L_\infty$-morphism respectively. As we will explain, classical closed string Maurer-Cartan elements (closed string backgrounds) modulo closed string gauge transformations, are in  one-to-one correspondence with classically consistent open string field theories modulo gauge transformations, which include open string background transformations as well as open string field redefinitions. Thus a closed string background not only determines a unique closed string field theory but also a unique classical open SFT, modulo gauge transformations.

We will show that the latter isomorphism persists at the quantum level although the complete quantum closed string  Maurer-Cartan equation will generically have no solutions, which is a reflection of the fact that a SFT of just open strings is quantum mechanically incomplete. The exception to this is when the closed string symplectic structure is degenerate on shell, i.e. on the cohomology of the closed string BRST operator. This is one  of the distinguishing features of the topological string. In the latter case the Maurer-Cartan equation decomposes into two irreducible parts: an equation for the background and linear equation for the propagator.

 We should also emphasize the relevance of the open-closed correspondence  in establishing background independence  of closed string field theory described above. The isomorphism just described is instrumental in establishing background independence within the class of backgrounds that preserve the vector space of perturbative fluctuations. The details of this will be explained in the text.

\section{The Homotopy Algebra of String Field Theory}\label{sec:infty}

String vertices represent subspaces, i.e. singular chains, of the moduli space of Riemann surfaces. The corresponding chain complex admits the structure of a
BV algebra \cite{Zwiebach closed, Zwiebach open-closed}. The basic requirement for any SFT, that it reproduces the S-matrix amplitudes of perturbative string theory, translates into the statement 
that the singular chains defining the string vertices satisfy the BV master equation. This is the background independent data of SFT \cite{ZS background1,ZS background2}. A string background determines a world sheet  conformal field theory where the state space $A$ of this CFT (or a certain restriction thereof) is equipped with an odd symplectic structure $\omega$. This in turn makes the space $C(A)$ of functions on $A$ (the space of multilinear maps on $A$ with suitable symmetry properties) a BV algebra. The world sheet CFT  defines a morphism of BV algebras which implies that the 
BV master equation is also satisfied at the level of $C(A)$ \cite{Zwiebach closed, Zwiebach open-closed}.  

The most general theory involves open and closed strings and we have to consider the moduli spaces $\mathcal{P}^{b,g}_{n,m}$ \cite{Zwiebach open-closed}, where $g$ is the genus, $b$ is the number of boundaries, $n$ is the number of closed string punctures and $m=(m_1,\dots,m_b)$ where $m_i$ is the number of open string punctures on the $i$-th boundary. Furthermore, the geometric vertices which we will denote by $\mathcal{V}^{b,g}_{n,m}\subset \mathcal{P}^{b,g}_{n,m}$, have to be invariant under the following transformations:
\bi
\item[(i)]cyclic permutation of open string punctures on one boundary
\item[(ii)]arbitrary permutation of closed string punctures
\item[(iii)]arbitrary permutation of boundaries
\ei
Consider now a fixed background, that defines a world sheet CFT. The corresponding state space of open strings is denoted by $A_o$ and the restricted state space of closed strings (those states annihilated by $b_0^-$ and $L_0^-$) by $A_c$. We use the conventions where the string fields have degree zero, both in the closed string and the open string sector \cite{Zwiebach open-closed,qocha}. The world sheet CFT preserves the above symmetry properties, that is
\be\no
 \mathcal{P}^{b,g}_{n,m}\supset \mathcal{V}^{b,g}_{n,m}\mapsto f^{b,g}_{n,m}\in \op{Hom}\bracketi{A_c^{\w n}\tp (A_o^{\tp m_1})^{cycl}\w \dots \w (A_o^{\tp m_b})^{cycl} ,R}\;\text{,}
\ee
where $\wedge$ denotes the graded symmetric product and $R$ is the module of commuting and anti-commuting numbers. The maps $f^{b,g}_{n,m}$ are the algebraic string vertices. In the following we will usually not distinguish between algebraic and geometric vertices, whenever the meaning is clear from the context. The string field theory action for the open string field $a\in A_o$ and the closed string field $c\in A_c$ is then given by the sum of all string vertices, weighted with appropriate powers of $\hbar$ and symmetry factors \cite{Zwiebach open-closed}:
\be\label{eq:ocaction}
S(c,a)=\sum_{b,g}\sum_{n,m}\inv{b!} \inv{n!}\inv{m_1\dots m_b} \,\hbar^{2g+b+n/2-1}\,f^{b,g}_{n,m}\bracketi{c^{\w n};a^{\tp m_1},\dots,a^{\tp m_b}} \;\text{.}
\ee
The quantum BV master equation reads 
\be\label{eq:BVeq}
\hbar\Delta^{BV}S+\inv{2}(S,S) =0\;\text{,}
\ee
where $\Delta^{BV}$ is the BV operator induced by the odd symplectic structure $\omega$ (bpz inner product) on the state space of the world sheet CFT, and $(\cdot,\cdot)$ is the associated odd Poisson bracket (antibracket) \cite{qocha, Schwarz bv}. Since the odd symplectic structure splits into open and closed parts $\omega=\omega_o+\omega_c$, the BV operator and the odd Poisson bracket split as well:
\be\no
\Delta^{BV}=\Delta^{BV}_o+\Delta^{BV}_c \sep (\cdot,\cdot)=(\cdot,\cdot)_o+(\cdot,\cdot)_c \;\text{.}
\ee
The geometric counterpart of $\Delta^{BV}_o$ and $\Delta^{BV}_c$ at the level of chain complexes of moduli spaces is the sewing of open and closed string punctures, respectively. 
The homotopy algebra corresponding to that full-blown theory is the quantum open-closed homotopy algebra (QOCHA) \cite{qocha}, but there are many sub-algebras corresponding to certain limits of this theory, which will be discussed in the following.

\subsubsection{Classical Theory}
Let us consider the limit where we restrict to those moduli spaces that are closed under sewing at tree level. For open SFT the relevant surfaces are discs with punctures on the boundary, whereas in closed SFT we have to consider punctured spheres. 

\begin{figure}[h] \centering
\begin{minipage}[b]{0.4\textwidth}
classical open SFT\\\vspace{2mm}
\resizebox{2.5cm}{!}{
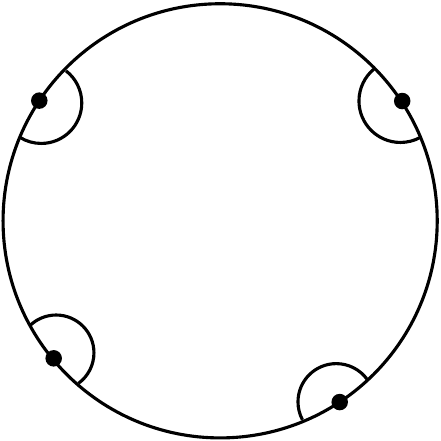}
\end{minipage}
\begin{minipage}[b]{0.4\textwidth}
classical closed SFT\\\vspace{2mm}
\resizebox{2.5cm}{!}{
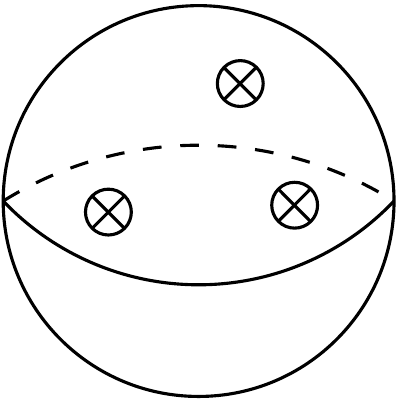}
\end{minipage}
\caption{punctured disc/sphere} 
\end{figure}
Similarly such a theory satisfies the classical BV master equation
\be\no
(S,S)=0 \;\text{.}
\ee

In classical open SFT, the action thus reads (see equation (\ref{eq:ocaction}))
\be\no
S(a)=\sum_n \inv{n} f^{1,0}_{0,n}(a^{\tp n}) \;\text{,}
\ee
and the classical BV master equation implies that the multilinear maps  $m_n:A^{\tp n}_o\to A_o$ defined by
\be\no
\omega_o(m_n\,,\,\cdot\,)\defineL  f^{1,0}_{0,n+1}
\ee
satisfy the relations of an $A_\infty$-algebra\footnote{An $A_\infty$-algebra actually corresponds to the case of a single D-brane. For several D-branes, one obtains a Calabi-Yau $A_\infty$ category (See for example \cite{Costello, Chen liebi}).} \cite{Zwiebach open1}. Similarly the multilinear maps $l_n:A_c^{\w n}\to A_c$ associated to the classical action $S(c)$ of closed SFT (after absorbing $\hbar^{1/2}$ in the closed string field)
\be\no
S(c)=\sum_n \inv{n!} f^{0,0}_{n,0}(c^{\w n})  \sep \omega_c(l_n\,,\,\cdot\,)\defineL  f^{0,0}_{n+1,0} \;\text{,}
\ee
obey the relations of a $L_\infty$-algebra \cite{Zwiebach closed}.

Finally, there is also a sub-algebra corresponding to a theory of open and closed strings. We consider spheres with closed string punctures, discs with open string punctures and additionally discs with open and closed punctures.

\begin{figure}[h] \centering
\begin{minipage}[b]{0.3\textwidth}
\resizebox{2.5cm}{!}{
\input{sphere.pdf_t}}
\end{minipage}
\begin{minipage}[b]{0.3\textwidth}
\resizebox{2.5cm}{!}{
\input{discopen.pdf_t}}
\end{minipage}
\begin{minipage}[b]{0.3\textwidth}
\resizebox{2.5cm}{!}{
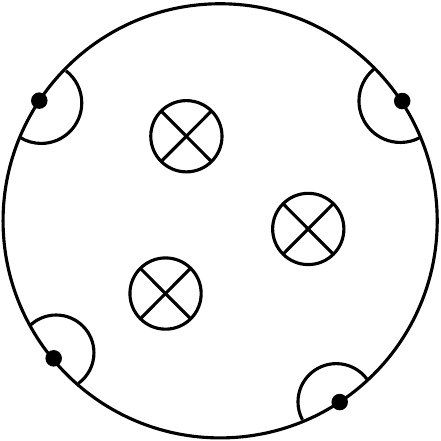}
\end{minipage}
\caption{surfaces of 'classical' open-closed SFT}
\end{figure}

In order to make this theory well defined, we have to exclude the operation of sewing a closed string puncture on one disc to another closed string puncture on a second disc. This would produce 
surfaces with more than one boundary, i.e. surfaces which are not part of the theory. Physically speaking, we treat the closed string as an external field. After absorbing $\hbar^{1/2}$ in the closed string field, the action reads
\be\no
S(c,a)=\inv{\hbar}\sum_n \inv{n!} f^{0,0}_{n,0}(c^{\w n})+\sum_n \inv{n} f^{1,0}_{0,n}(a^{\tp n})+\sum_{n,m} \inv{n!}\inv{m} f^{1,0}_{n,m}(c^{\w n};a^{\tp m}) \;\text{,}
\ee 
and satisfies the classical BV master equation to order $\hbar^{0}$ (Note that the closed string Poisson bracket is proportional to $\hbar$ in this normalization.). Translated into the language of homotopy algebras we get the following: 
Let's define multilinear maps $n_{n,m}:A_c^{\w n}\tp A_o^{\tp m}\to A_o$ associated to discs with open and closed punctures by
\be\no
\omega_o(n_{n,m}\,,\,\cdot\,)\defineL f^{1,0}_{n,m+1} \;\text{.}
\ee
Furthermore we collect the individual maps to
\begin{align}\no
l&\defineL \sum_n l_n \;:\;SA_c\to A_c \\\no
m&\defineL \sum_n m_n \;:\;TA_o\to A_o \\\no
n&\defineL \sum_{n,m} n_{n,m} \;:\; SA_c\tp TA_o\to A_o \;\text{,}
\end{align}
where $TA$ and $SA$ denote the tensor algebra and the graded symmetric tensor algebra respectively. To the first two maps we can associate a coderivation (see appendix \ref{app1} for details about coderivations and homotopy algebras). That is,
\begin{align}\no
L&\defineL \liftd{l} \;\in \op{Coder}^{cycl}(SA_c) \\\no
M&\defineL \liftd{m} \;\in \op{Coder}^{cycl}(TA_o) \\\no
N&\defineL \liftd{n} \;:\; SA_c \to \op{Coder}^{cycl}(TA_o) \;\text{,}
\end{align}
where  the map $N$, associated to discs with open and closed punctures, induces an $L_\infty$-morphism from the $L_\infty$-algebra $(A_c,L)$ of closed strings  to the differential graded Lie algebra $(\op{Coder}^{cycl}(TA_o),d_h,[\cdot,\cdot])$ which controls deformations of the open string field theory $(A_o,M)$ \cite{Kajiura open-closed1,Kajiura open-closed2}: 
\be\label{eq:ocha}
(A_c,L)\xrightarrow[]{L_\infty-\text{morphism}} (\op{Coder}^{cycl}(TA_o),d_h,[\cdot,\cdot])\;\text{}
\ee
More precisely, we have
\be\label{eq:ochaex}
N\circ L= d_h\circ N +\inv{2} [N,N]\circ \Delta \;\text{,}
\ee
where $\Delta$ denotes the comultiplication in $SA_c$. This algebra is called open-closed homotopy algebra (OCHA) \cite{Kajiura open-closed1,Kajiura open-closed2} and will be essential in section \ref{sec:oc}.

\subsubsection{Quantum Theory}
At the quantum level there there is no consistent open SFT, since e.g. open string one-loop diagrams can be interpreted as closed string tree-level amplitudes. 
For a theory of closed strings we have to consider surfaces of arbitrary genus with an arbitrary number of punctures, and the action according to equation (\ref{eq:ocaction}) reads (after absorbing appropriate powers of $\hbar $)
\be\no
S(c)=\sum_g \sum_n \frac{\hbar^{g}}{n!} \;f^{0,g}_{n,0}(c^{\w n}) \;\text{.}
\ee
We define multilinear maps $ l^{g}_n:A_c^{\w n}\to A_c$ via
\be\no
\omega_c(l^g_n \,,\, \cdot)\defineL f^{0,g}_{n+1,0} \;\text{,}
\ee
and lift $l^g=\sum_n l^g_n$ to a coderivation
\be\no
L^g\defineL \liftd{l}^g \in \op{Coder}^{cycl}(SA_c) \;\text{.}
\ee
The closed string BV operator $\Delta^{BV}_c$ requires the inclusion of a second order coderivation $\Omega_c^{-1}$, which is defined to be the lift of  the inverse of the odd symplectic structure $\omega_c$:
\be\no
\Omega_c^{-1}\defineL \widehat{\omega_c^{-1}} \in \op{Coder}^2(SA_c)\;\text{.}
\ee
The main identity of closed string field theory \cite{Zwiebach closed} together with the cyclicity condition is equivalent to the statement that $\frak{L}_c\in \coder(SA_c,\hbar)$ defined by
\be\label{eq:loop}
\frak{L}_c\defineL \sum_g \hbar^g L^g + \hbar\,\Omega_c^{-1} \;\text{,}
\ee
squares to zero \cite{Markl loop}. This algebra is called loop homotopy algebra, which is obviously a special case of an $IBL_\infty$-algebra (see appendix \ref{app1}). 

\begin{figure}[h] \centering
\begin{minipage}[t]{0.4\textwidth}
closed SFT\\\vspace{2mm}
\resizebox{5cm}{!}{
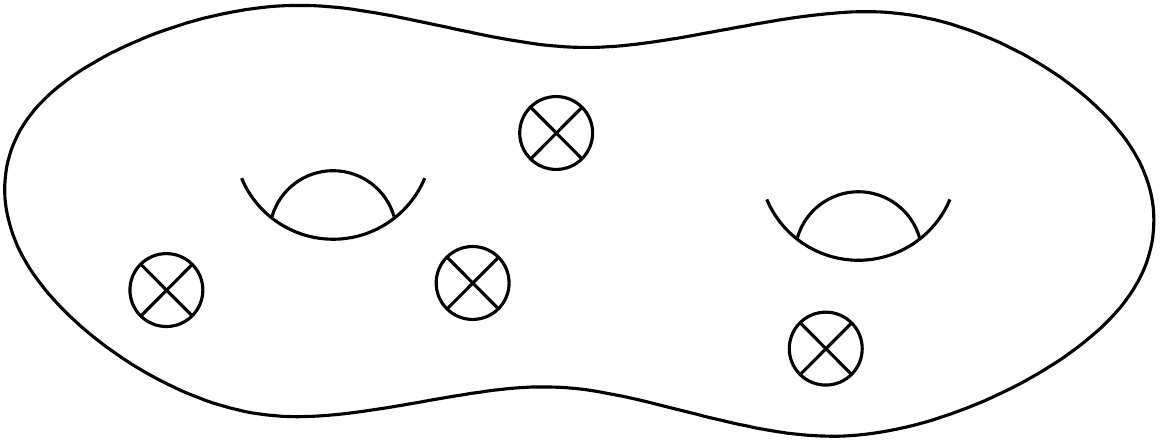}
\end{minipage}
\begin{minipage}[t]{0.4\textwidth}
open-closed SFT\\\vspace{2mm}
\resizebox{5cm}{!}{
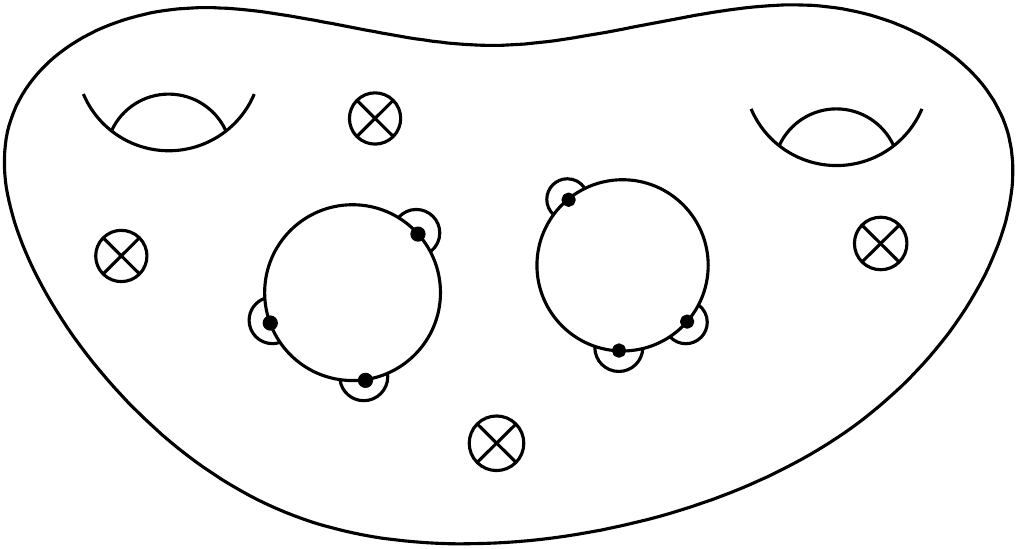}
\end{minipage}
\caption{surfaces in quantum SFT}
\end{figure}

The algebraic structure of quantum open-closed SFT can be described in a similar way as in 'classical' open-closed SFT. The surfaces with open and closed string punctures define a morphism, but in this case an $IBL_\infty$-morphism rather than a $L_\infty$-morphism. On the open string side of the OCHA (\ref{eq:ocha}) we had the differential graded Lie algebra $(\op{Coder}^{cycl}(TA_o),d_h,[\cdot,\cdot])$, but note that due to the isomorphism $\op{Coder}^{cycl}(TA_o)\stackrel{\pi_1}\cong \op{Hom}^{cycl}(TA_o,A_o)\stackrel{\mathrm{\omega_o}}\cong \op{Hom}^{cycl}(TA_o,R)$ the Hochschild differential $d_h$ and the Gerstenhaber bracket $[\cdot,\cdot]$ have their counterparts defined on $\mathcal{A}_o\defineL \op{Hom}^{cycl}(TA_o,R)$ (see e.g. \cite{qocha} for more details), which we will also denote by 
\be\no
d_h: \mathcal{A}_o\to \mathcal{A}_o 
\ee
and 
\be\no
[\cdot,\cdot]: \mathcal{A}_o^{\w 2}\to \mathcal{A}_o \;\text{.}
\ee
(The Gerstenhaber bracket is now symmetric in the inputs and has degree one, since $\omega_o$ has degree minus one.) In the following we will work with the space $\mathcal{A}_o$, which is called the cyclic Hochschild complex, rather than with $\op{Coder}^{cycl}(TA_o)$. In order to take account of the open string BV operator $\Delta^{BV}_o$, we have to supplement the differential graded Lie algebra $(\mathcal{A}_o,d_h,[\cdot,\cdot])$ by an additional operation
\be\no
\delta:\mathcal{A}_o \to \mathcal{A}_o^{\w 2} \;\text{,}
\ee
defined by 
\begin{align}\label{eq:delta}
(\delta f)&(a_1,\dots,a_n)(b_1,\dots,b_m)\\\no
\defineL&(-1)^f\sum_{i=1}^{n}\sum_{j=1}^{m}(-1)^\epsilon f(e_k,a_{i},\dots,a_n,a_1,\dots,a_{i-1},e^k,b_{j},\dots,b_m,b_1,\dots,b_{j-1})\;\text{,}
\end{align}
where $(-1)^\epsilon$ denotes the Koszul sign, $\{e_i\}$ is a basis of $A_o$ and $\{e^i\}$ is the corresponding dual basis satisfying $\omega_o({}_ie,e^j)={}_i\delta^j$ (see \cite{Witt super,qocha} for the sign conventions for left and right indices). In \cite{Chen liebi,Cieliebak ibl} it has been shown that $(\mathcal{A}_o,d_h,[\cdot,\cdot],\delta)$ defines an involutive Lie bialgebra, a special case of an $IBL_\infty$-algebra. In the language of $IBL_\infty$-algebras this is equivalent to the statement that
\be\no
\LL_o\defineL \liftd{d}_h+\widehat{[\cdot,\cdot]}+\hbar\, \liftd{\delta}\; \in \coder(\mathcal{A}_o,\hbar)
\ee
squares to zero. Now for $b\ge1$ and $g\ge0$, we define maps $n^{b,g}\in \op{Hom}(SA_c,\mathcal{A}_o^{\w b})$ by
\be\no
n^{b,g}=\begin{cases}\sum_{n=1}^\infty\sum_m f^{1,0}_{n,m} & \text{,}\; b=1,\;g=0 \\ \sum_{n=0}^\infty\sum_m f^{b,g}_{n,m} &\text{,}\; \text{else}\end{cases}  \;\text{.}
\ee
We exclude $f^{1,0}_{0,n}$ in the sum for $b=1$, $g=0$, since it is already taken into account via the Hochschild differential $d_h$. 
Finally, the algebraic structure of quantum open-closed SFT can be summarized in the following way:
The open-closed vertices $n^{b,g}$ define an $IBL_\infty$-morphism from the loop homotopy algebra of closed strings $A_c$ to the involutive Lie bialgebra on the cyclic Hochschild complex of open strings $\mathcal{A}_o$
\be\no
(A_c,\LL_c)\xrightarrow[]{IBL_\infty-\text{morphism}} (\mathcal{A}_o,{\LL}_o) \;\text{.}
\ee
That is we have 
\be\label{eq:qocha1}
\lifte{\frak{n}}\circ \LL_c=\LL_o\circ\lifte{\frak{n}}
\ee
where 
\be\no
\frak{n}=\sum_{b=1}^\infty\sum_{g=0}^\infty \hbar^{b+g-1}\, n^{b,g} \;\text{.}
\ee
This is the quantum open-closed homotopy algebra introduced in \cite{qocha}. 
Equation (\ref{eq:qocha1}) can be recast, such that the five distinct sewing operations in open-closed SFT become apparent \cite{qocha}:
\begin{align}\label{eq:qocha2}
\frak{n}&\circ \LL_c + \frac{\hbar}{2}\bracketi{\frak{n}\circ \liftd{e}_i\w \frak{n}\circ \liftd{e}^i}\circ\Delta \\\no
&= \LL_o \circ \frak{n} +\inv{2}\widehat{[\cdot,\cdot]}\circ \bracketi{\frak{n}\w \frak{n}}\circ\Delta -\bracketi{(\widehat{[\cdot,\cdot]}\circ\frak{n})\w\frak{n}}\circ\Delta \;\text{.}
\end{align}
In equation (\ref{eq:qocha2}), $e_i$ and $e^i$ denote a basis and corresponding dual basis of $A_c$ w.r.t. the symplectic structure $\omega_c$. Obviously we recover the OCHA of equation 
(\ref{eq:ochaex}) in the limit $\hbar\to 0$.

In \cite{Markl loop} it has been shown that the closed string loop homotopy algebra (\ref{eq:loop}) defines an algebra over the Feynman transform of $\text{Mod}(Com)$.
Similarly, it is expected that the QOCHA of open-closed SFT actually describes an algebra over the Feynman transform of a (two colored) operad corresponding to the moduli spaces of \cite{Liu OCmoduli}.
For more information in this direction see \cite{Voronov BVopenclosed, Harrelson openclosed}.

\section{Decomposition Theorem for closed String Loop Algebra}\label{sec:dec}
In the previous section we reformulated the BV master equation for the string vertices as axioms of some homotopy algebra. The connection between the S-matrix of SFT and the perturbative string amplitudes is then established via the minimal model theorem. Consider for example classical open SFT, and denote its corresponding $A_\infty$-algebra by $(A,M)$. The minimal model theorem states that the cohomology $H(A,d)$ of $A$ with respect to the differential $d=\pi_1\circ M \circ i_1$ admits the structure of an $A_\infty$-algebra, denoted by $(H(A,d),\tilde{M})$, with vanishing differential $\pi_1\circ \tilde{M}\circ i_1=0$. Furthermore, $(H(A,d),\tilde{M})$ and $(A,M)$ are quasi-isomorphic, i.e. there is an $A_\infty$-quasi-isomorphism $\tilde{F}:(H(A,d),\tilde{M})\to(A,M)$. Note that in SFT the differential $d$ is the BRST operator and the BRST cohomology $H(A,d)$ represents the physical states. 

The construction of the minimal model is, in fact, identical to the construction of tree level S-matrix amplitudes via Feynman rules: First one chooses a certain gauge, such that we can define a propagator. With the aid of the propagator we construct all possible trees with vertices labeled $m_n\defineL \pi_1\circ M \circ i_n$ and internal lines labeled by the propagator. The collection of all these trees, restricted to the cohomology $H(A,d)$, then defines the multilinear maps $\tilde{m}=\pi_1\circ\tilde{M}$. Thus $\tilde{m}$ represents the S-matrix amplitudes, and moreover the $A_\infty$ relations for the S-matrix elements, $\tilde{M}^2=0$, can be identified as the Ward identities. 

The relation between the minimal model and S-matrix amplitudes in classical open SFT is discussed in \cite{Kajiura open1,Kajiura open2}. In classical closed SFT, the algebraic structure induced by the S-matrix elements on the BRST cohomology is accordingly that of an $L_\infty$-algebra \cite{ZwiebachWitten}, and the minimal model in the context of $L_\infty$-algebras is discussed in \cite{Fukaya,Kontsevich}. Furthermore there is a generalization of the minimal model theorem in the form of the decomposition theorem, which states that an $A_\infty$/$L_\infty$-algebra is isomorphic to the direct sum of a linear contractible part and a minimal part \cite{Kontsevich, Kajiura open1, Kajiura open2}.

In this section we are concerned with analogous statements in quantum closed SFT. The Ward identities of quantum closed SFT can be interpreted as the loop homotopy algebra axioms \cite{Zwiebach closed,Verlinde}. In chapter \ref{sec:infty}, we pointed out that loop homotopy algebras are indeed algebras over the Feynman transform of a modular operad \cite{Markl loop}, and the minimal model theorem corresponding to such algebras has been established in \cite{Lazarev1, Lazarev2}. The explicit construction of such minimal models resembles that in the case of $A_\infty$-algebras, but where one has to consider graphs (allowing loops) instead of trees. 

In the first subsection we will review what kind of extra structure is needed in order to define the minimal model/decomposition model, and the relation of these extra structures to the notion of gauge fixing in SFT. The second subsection is devoted to the proof of the decomposition theorem for loop homotopy algebras and finally we derive thereof the minimal model theorem. Indeed we will need the decomposition theorem, rather than the minimal model theorem, for the considerations in section \ref{sec:uni}. Besides an explicit construction of the decomposition model, we also give an explicit construction of the $IBL_\infty$-isomorphism from the initial loop homotopy algebra to its decomposition model.

\subsection{Hodge decomposition and gauge fixing}
Let $A$ be a graded module endowed with an odd symplectic structure $\omega$ of degree minus one and a compatible differential $d:A\to A$ of degree one, i.e.
\be\no
d^2=0 \mand \omega(d,\,\cdot\,)+\omega(\,\cdot\,,d)=0 \;\text{.}
\ee
\begin{Def}
A pre Hodge decomposition of $A$ is a map $h:A\to A$ of degree minus one which is compatible with the symplectic structure and squares to zero. 
\end{Def}
For a given pre Hodge decomposition of $A$, we define the map 
\be\no
P=1+dh+hd \;\text{,}
\ee
which obviously satisfies $Pd=dP$ and $Ph=hP$. 
\begin{Def}
A Hodge decomposition of $A$ is a pre Hodge decomposition which additionally satisfies $hdh=-h$.
\end{Def}
Let $h$ be a Hodge decomposition of $A$ and define $P_U=-hd$ and $P_T=-dh$. Then the following properties are satisfied:
\be\no
P^2=P \sep
P_U^2=P_U \sep
P_T^2=P_T\;\text{.}
\ee
That is $P,P_U,P_T$ are projection maps and $A$ decomposes into the corresponding projection subspaces $A_P\oplus A_U \oplus A_T$. Furthermore we have $Ph=hP=0$.
\begin{Def}
A Hodge decomposition of $A$ is called harmonious if $dhd=-d$.
\end{Def}
For a harmonious Hodge decomposition the additional feature compared to a Hodge decomposition is $Pd=dP=0$. Furthermore we have $A_P\perp A_U\oplus A_T$, $A_U\perp A_U$ and $A_T\perp A_T$.
These definitions are borrowed from \cite{Lazarev2}.

Let us now elucidate how the algebraic structures just described come into play in SFT. Let $d$ be the BRST differential and $A$ the space of string fields. Gauge fixing is required to obtain a well defined path integral, which amounts to fixing a representative for every element of the cohomology $H(A,d)$. More precisely, the gauge fixing determines a map
\be\no
i:H(A,d)\to A \; \text{,}
\ee
which maps an element of the cohomology to its corresponding representative. We will call $i$ the inclusion map.
We also have the projection map 
\be\no
\pi: A\to H(A,d) \;\text{.}
\ee
Obviously, the map $P\defineL i\circ\pi:A\to A$ satisfies $P^2=P$ and the image $A_P$ of $P$ is isomorphic to $H(A,d)$. That is $A_P$ represents the physical states. 
Moreover $P$ is a chain map, i.e. $Pd=dP=0$, and its induced map on cohomology is the identity map. This implies that $P$ is homotopic to $1$, i.e. there is a map $h:A\to A$ of degree one such that
\be\no
P-1 = hd+dh\;\text{.}
\ee 
Note that $P^2=P$ implies $h^2=0$. Physically we can identify $h$ as the propagator corresponding to the chosen gauge. We demand $hP=Ph=0$, which means that we set the propagator to zero on the space of physical states. The subspace $A_U$ corresponding to the projection map $P_U=-hd$ represents the unphysical states, i.e. the states not annihilated by $d$, and the subspace $A_T$ represents the space of trivial states, i.e. $d$ exact states. Thus we can summarize that choosing a gauge in SFT determines a harmonious Hodge decomposition, which decomposes the state space into physical, unphysical and trivial states \cite{Kajiura open1,Kajiura open2}. When dealing with a pre Hodge decomposition, we will call the images of $P$, $-dh$, $-hd$ the physical space, trivial space, unphysical space as well.

In the next subsection we will see that the extra data required to construct a decomposition model is just a pre Hodge decomposition, whereas we need a harmonious Hodge decomposition to construct a minimal model.
 
\subsection{Decomposition theorem of loop homotopy algebra}
Let $(A,\LL)$ be a loop homotopy algebra, i.e. 
\be\label{eq:loopII}
\LL=\sum \hbar^g L^g +\hbar \Omega^{-1} \;\text{,}
\ee
where $L^g=\liftd{l^g}\in \op{Coder}^{cycl}(SA)$ and $\Omega^{-1}=\liftd{\omega^{-1}}\in \op{Coder}^2(SA)$ is the lift of the inverse of the odd symplectic structure (see equation (\ref{eq:loop})). 
We define $l_q\defineL\sum_g \hbar^g l^g$ and $l_{cl}\defineL l^0$, where the subscripts indicate quantum and classical respectively. The differential on $A$ is given by $d=l_{cl}\circ i_1$. Furthermore we abbreviate the collection of multilinear maps without the differential by $l_q^\ast\defineL l_q-d$ and $l_{cl}^\ast\defineL l_{cl}-d$. 

In appendix \ref{app1} we introduced the lifting map, which lifts multilinear maps to a coderivations, but for notational convenience we will denote this map by $D$ rather than a hat in the following. With these conventions equation (\ref{eq:loopII}) reads
\be\no
\LL= D(d+ l_q^\ast + \hbar \omega^{-1}) \;\text{.}
\ee 
The loop homotopy algebra axioms are summarized by $\LL^2=0$ and can be recast to
\be\label{eq:mainid}
d\circ l^\ast_q+l^\ast_q\circ D(l^\ast_q)+l^\ast_q\circ D(d) + l^\ast_q\circ D(\hbar \omega^{-1})=0 \;\text{,}
\ee
and
\be\label{eq:cyclicity}
l^{\ast}_q\circ D(e_i)\w e^i =0 \;\text{,}
\ee
where $\{e_i\}$ and $\{e^i\}$ denote a basis and corresponding dual basis of $A$ w.r.t. the symplectic structure $\omega$, that is $\omega^{-1}=\inv{2}e_i\w e^i$.
Equation (\ref{eq:mainid}) is called the main identity \cite{Zwiebach closed,Markl loop} whereas equation (\ref{eq:cyclicity}) states cyclicity of the maps $l^\ast_q$, i.e. that $\omega(l^g_n,\,\cdot)$ enjoys full symmetry in all $n+1$ inputs (see appendix \ref{app1}). The cyclicity condition (\ref{eq:cyclicity}) is essentially saying that there is actually no distinction between outputs and inputs.

To construct a decomposition model of the loop homotopy algebra (\ref{eq:loopII}), we additionally need the data of a pre Hodge decomposition $h:A\to A$. Again we define $P=1+dh+hd$, and in addition we introduce 
\be\label{eq:metric}
g\defineL -\omega\circ d \mand g^{-1}\defineL h\circ\omega^{-1} \;\text{,}
\ee
where the symplectic structure $\omega$ and its inverse $\omega^{-1}$ are considered as a map from $A$ to $A^\ast$ and $A^\ast$ to $A$, respectively. Since $d$ and $h$ are compatible with the symplectic structure, $g$ is a symmetric map and $g^{-1}\in A^{\w 2}$, each of degree zero. Assume for a moment that $h$ defines indeed a harmonious Hodge decomposition, then we saw that the full space $A$ splits into $A_P\oplus A_U \oplus A_T$, where $A_P$, $A_U$, $A_T$ represents the physical, unphysical, trivial space respectively. In this case $g$ is non-vanishing only on the unphysical space $A_U$ and $g^{-1}$ defines its inverse upon restricting to $A_U$, that is $g$ defines a metric on the unphysical space. 

In the context of $L_\infty$-algebras the decomposition theorem is proven by constructing trees from $l_{cl}^\ast$ and $h$ \cite{Fukaya}. There is a nice way of generating these trees, by employing the tools developed in appendix \ref{app1} \cite{Kajiura open1, Kajiura open2}: Consider trees where the root and the internal lines are labelled by the propagator $h$, the vertices by $l_n=l_{cl}^\ast\circ i_n$ and the leaves by the identity map $1$. The collection of all these trees $\mathsf{T}_{cl}:SA\to A$, is defined recursively via
\be\label{eq:trees}
\mathsf{T}_{cl}=h\circ l^\ast_{cl}\circ e^{1+\mathsf{T}_{cl}} \mand \mathsf{T}_{cl}\circ i_1 =0 \;\text{,}
\ee
where $e$ is the lifting map of multilinear maps to cohomomorphisms (see appendix \ref{app1}).
In figure \ref{fig:trees}, we depict the first few terms of $\mathsf{T}_{cl}$ according to the number of inputs.  

\begin{figure}[h] \centering
\begin{minipage}[b]{0.7\textwidth}
\scalebox{0.9}{
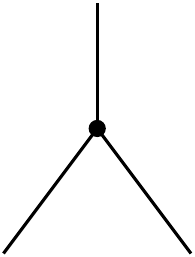}
\end{minipage}
\\
\vspace{0.6cm}
\begin{minipage}[b]{0.7\textwidth}
\scalebox{0.9}{
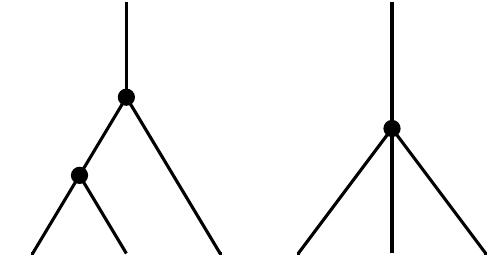}
\end{minipage}
\\
\vspace{0.6cm}
\begin{minipage}[b]{0.7\textwidth}
\scalebox{0.9}{
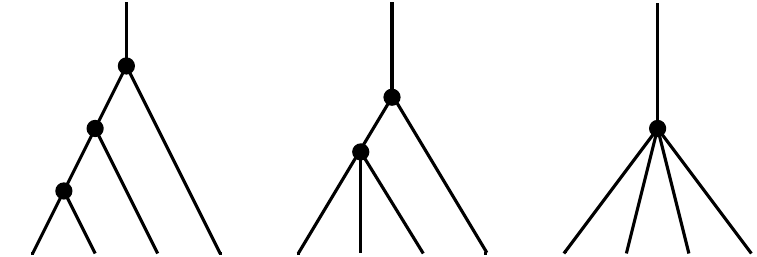}
\end{minipage}
\caption{Trees constructed from $l_{cl}^\ast$ and $h$, where $l_n=l_{cl}^\ast\circ i_n$. Although not explicitly indicated, the inputs are understood to be symmetrized.}
\label{fig:trees}
\end{figure}

Likewise, for an arbitrary linear map $x:A\to A$ we define trees in the same way as in equation (\ref{eq:trees}) but replacing the root by $x$, that is
\be\label{eq:treesroot}
\overset{(x)}{\mathsf{T}}_{cl}\defineL x\circ l^\ast_{cl}\circ e^{1+\mathsf{T}_{cl}} \;\text{.}
\ee

As anticipated in the beginning of this section, we have to consider graphs to prove the decomposition theorem for loop homotopy algebras. Graphs are essentially trees with loops attached. 
The strategy is thus to start with trees as in the $L_\infty$ case. Attaching loops can then be implemented neatly by composing the trees with an appropriate cohomomorphism. So let us first define trees constructed recursively from $l^\ast_q$ and $h$ via
\be\label{eq:treesq}
\mathsf{T}_{q}=h\circ l^\ast_{q}\circ e^{1+\mathsf{T}_{q}} \mand \mathsf{T}_{q}\circ i_1 =0 \;\text{.}
\ee
Consider now the cohomomorphism
\be\no
E(\hbar g^{-1})\defineL e^{1+\hbar g^{-1}} \in \cohom(SA,SA,\hbar) \;\text{,}
\ee
where $e$ is the lifting map and $g^{-1}$ is the inverse metric on the unphysical space defined in equation (\ref{eq:metric}).
Let $\{u_i\}$ be a basis of the unphysical space and $\{u^{i}\}$ its dual basis w.r.t. $g$, i.e. 
\be\no
g({}_iu,u^j)={}_i\delta^j \;\text{,}
\ee
where we use the sign conventions of \cite{qocha,Witt super} that relate left indexed objects with right indexed objects. 
In terms of basis and dual basis, we can express the inverse metric as
\be\no
g^{-1}=\inv{2}u_i\w {}^iu \;\text{.}
\ee
Physically $g^{-1}$ is interpreted as a loop, it connects two inputs by propagating the unphysical degrees of freedom.
The cohomomorphism $E(\hbar g^{-1})$ is then the map that attaches loops in all possible ways. Thus $E(\hbar g^{-1})$ is the map that we have to compose with 
the trees $\mathsf{T}_q$ to obtain graphs. Since we are actually interested in graphs with many outputs (directed connected graphs), we define the collection of all these graphs $\B{\Gamma}$ by
\be\label{eq:graphs}
e^{1+\hbar g^{-1} + \B{\Gamma}}= e^{1+\mathsf{T}_q}\circ E(\hbar g^{-1}) \;\text{.}
\ee 
Note that since $E(\hbar g^{-1})$ acts on the collection of disconnected trees $e^{1+\T_q}$, every $g^{-1}$ either generates a loop or increases the number of outputs by one.
In figure \ref{fig:graph} we depict a typical graph generated in that way.  Upon amputating the loops $g^{-1}$, every graph reduces to a collection of connected trees.

\begin{figure}[h] \centering
\begin{minipage}[b]{0.9\textwidth}
\scalebox{0.9}{
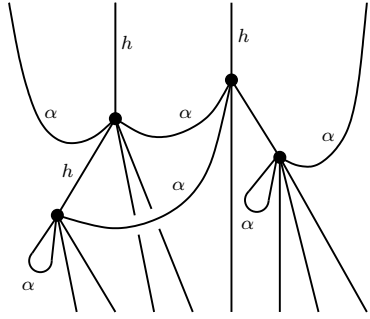}
\end{minipage}
\caption{Graph constructed from $h$, $l^\ast_q$ and $\alpha=g^{-1}$. The curved lines represent the inverse metric $g^{-1}$.}
\label{fig:graph}
\end{figure}

A graph comes with a certain power of $\hbar$, which is the number of outputs minus one plus the number of loops (the first Betti number of the graph) plus the powers of $\hbar$ from the vertices.
The number of loops plus the powers of $\hbar$ from the vertices define the genus of the graph. Thus we have $\B{\Gamma}\in \bigoplus_{n=1}^\infty \hbar^{n-1}\op{Hom}(SA,\Sigma^n A)$, with
\be\no
\B{\Gamma}=\sum_{n=1}^\infty\sum_{g=0}^\infty \hbar^{n+g-1} \B{\Gamma}^{n,g} \;\text{,}
\ee
where $\B{\Gamma}^{n,g}$ represents the collection of graphs of genus $g$ with $n$ outputs and 
\be\label{eq:decmorph}
\bar{\FF}\defineL e^{1+\hbar g^{-1}+\B{\Gamma}} \in \cohom(SA,SA,\hbar)\;\text{.}
\ee
Now we are ready to state the decomposition theorem for loop homotopy algebras.

\begin{Thm}\label{thm:dec}
Let $(A,\LL=D(d+l^\ast_q +\hbar \omega^{-1}))$ be a loop homotopy algebra. For a given pre Hodge decomposition $h$, there is an associated loop homotopy algebra 
\be\label{eq:deccoder}
\bar{\LL}=D(d+ \overset{(P)}{\mathsf{T}}_q\circ E(\hbar g^{-1}) +\hbar \bar{\omega}^{-1}) \;\text{,}
\ee 
where $\bar{\omega}^{-1}=P^{\w2}(\omega^{-1})$ is the restriction of $\omega^{-1}$ to the physical space and $\overset{(P)}{\mathsf{T}}_q\circ E(\hbar g^{-1})$ represents the graphs with a single output labeled by $P$. Furthermore $\bar{\FF}=e^{1+\hbar g^{-1}+\B{\Gamma}}$ (see equation (\ref{eq:decmorph})) defines an $IBL_\infty$-isomorphism from $(A,\bar{\LL})$ to $(A,\LL)$. $d$ is called the linear contractible part and $\overset{(P)}{\mathsf{T}}_q\circ E(\hbar g^{-1}) +\hbar \bar{\omega}^{-1}$ the minimal part.
\end{Thm}

\begin{proof}
Since we expressed graphs as the composition of trees with the cohomomorphism $E(\hbar g^{-1})$, the proof can be traced back to the level of trees. In the following we will leave out the subscript $q$, i.e. $l^\ast=l^\ast_q$ and $\T_q=\T$. In a first step we show
\be\label{eq:treesid}
\overset{(hd)}{\T}+\T\circ D(d)+\T\circ D(\overset{(P)}{\T})+\T\circ D^{con}(\hbar \omega^{-1})=0 \;\text{,}
\ee
where $\T\circ D^{con}(\hbar \omega^{-1})$ means that we consider only those terms where $\omega^{-1}$ acts on one vertex and not those where $\omega^{-1}$ connects two vertices. We prove equation (\ref{eq:treesid}) inductively: Using the main identity (\ref{eq:mainid}) and $P=1+dh+hd$, we get
\begin{align}\label{eq:treesit}
\overset{(hd)}{\T}&=h\circ d\circ l^\ast\circ F = -h\circ l^\ast \circ D(d+l^\ast+ \hbar \omega^{-1})\circ F \\\no
&= -h\circ l^\ast \circ \bracketii{\bracketi{d+d\circ\T + l^\ast\circ F +\hbar \omega^{-1}}\w F}\circ \Delta \\ \no
&=-h\circ l^\ast \circ \bracketii{\bracketi{d+ \overset{(dh)}{\T}+\overset{(1)}{\T}+\hbar\omega^{-1}}\w F}\circ \Delta \\ \no
&=-h\circ l^\ast \circ \bracketii{\bracketi{d+ \overset{(P)}{\T}+\hbar \omega^{-1}- \overset{(hd)}{\T}}\w F}\circ \Delta  \;\text{,}
\end{align}
where $F=e^{1+\T}$. In equation (\ref{eq:treesit}), we use identities like (see appendix \ref{app1} for more details)
\be\no
D(l^\ast)\circ F=(l^\ast \w \op{id})\circ \Delta \circ F=(l^\ast\circ F \w F)\circ \Delta \;\text{.}
\ee
Upon iterating equation (\ref{eq:treesit}), we finally obtain equation (\ref{eq:treesid}).

Our strategy will be to first show $\LL\circ \bar{\FF}=\bar{\FF}\circ\bar{\LL}$. We start with calculating the left hand side:
\begin{align}\label{eq:left}
\LL\circ \bar{\FF}&=D(d)\circ F\circ E(\hbar g^{-1})+ D(l^\ast)\circ F\circ E(\hbar g^{-1})+D(\hbar \omega^{-1})\circ F\circ E(\hbar g^{-1}) \\\no
&=\bracketii{\bracketi{d+\overset{(dh)}{\T}+\overset{(1)}{\T}+\hbar \omega^{-1}}\w F}\circ \Delta \circ E(\hbar g^{-1}) \\\no
&=\bracketii{\bracketi{d+\hbar \omega^{-1}+\overset{(P)}{\T}-\overset{(hd)}{\T}}\w F}\circ \Delta \circ E(\hbar g^{-1}) 
\end{align}
On the right hand side we have:
\be\label{eq:right}
\bar{\FF}\circ \bar{\LL}=F\circ E(\hbar g^{-1})\circ D(d+\overset{(P)}{\T}\circ E(\hbar g^{-1})+\hbar\bar{\omega}^{-1})
\ee
Let us consider the individual terms one by one. From the definition of $g^{-1}$ it follows that 
\be\no
D(d)(g^{-1})=d(d_i)\w {}^id=\bar{\omega}^{-1}-\omega^{-1} \;\text{,}
\ee
and thus 
\be\no
E(\hbar g^{-1})\circ D(d)= D(d)\circ E(\hbar g^{-1})+ D(\hbar \omega^{-1})\circ E(\hbar g^{-1})-D(\hbar \bar{\omega}^{-1})\circ E(\hbar g^{-1})\;\text{.}
\ee
Therefore the first plus the third term of equation (\ref{eq:right}) yield
\begin{align}\no
F&\circ E(\hbar g^{-1})\circ D(d+\hbar\bar{\omega}^{-1})=F\circ D(d+\hbar \omega^{-1})\circ E(\hbar g^{-1})\\\no
&=\bracketii{\bracketi{d+\T\circ D(d)+ \hbar \omega^{-1}+\hbar e_i\w\T\circ D(e^i)+ \frac{\hbar}{2}\T\circ D(e_i)\w \T\circ D(e^i)  \\\no
& \hspace{9mm}+ \T\circ D(\hbar \omega^{-1})}\w F}\circ \Delta \circ E(\hbar g^{-1}) \;\text{.}
\end{align}
Using the cyclicity property (\ref{eq:cyclicity}), we conclude that 
\be\no
\hbar e_i\w\T\circ D(e^i)= \frac{\hbar}{2}\T\circ D(e_i)\w \T\circ D(e^i)=0 \;\text{,}
\ee
and
\be\no
\T\circ D(\hbar \omega^{-1})=\T\circ D^{con}(\hbar \omega^{-1})\;\text{.}
\ee
Similarly, cyclicity implies that the second term of equation (\ref{eq:right}) reduces to
\begin{align}\no
F\circ& E(\hbar g^{-1})\circ D(\overset{(P)}{\T}\circ E(\hbar g^{-1}))= F\circ D(\overset{(P)}{\T})\circ E(\hbar g^{-1})\\\no
&=\bracketii{\bracketi{\overset{(P)}{\T} + \T\circ D(\overset{(P)}{\T})}\w F}\circ\Delta \circ E(\hbar g^{-1})\;\text{.}
\end{align}
Altogether we finally get
\be\no
\bar{\FF}\circ\bar{\LL}=\bracketii{\bracketi{d+\hbar \omega^{-1}+\overset{(P)}{\T} +\T\circ D(d)+ \T\circ D(\overset{(P)}{\T})+\T\circ D^{con}(\hbar\omega^{-1})}\w F}\circ\Delta \circ E(\hbar g^{-1})\;\text{,}
\ee
and $\LL\circ \bar{\FF}=\bar{\FF}\circ\bar{\LL}$ follows then directly form equation (\ref{eq:treesid}).

The second part of the proof, $\bar{\LL}^2=0$, follows directly from $\LL\circ \bar{\FF}=\bar{\FF}\circ\bar{\LL}$. Note that $\bar{\FF}$ is an $IBL_\infty$-isomorphism, which implies that there is a unique inverse $\bar{\FF}^{-1}$. Thus we have
\be\no
\bar{\LL}^2=\bar{\FF}^{-1}\circ \LL^2\circ \bar{\FF}=0 \;\text{.}
\ee
\end{proof}

\subsection{Minimal model of loop homotopy algebra}
The minimal model theorem follows readily from the decomposition theorem, but in contrast to the decomposition theorem we need a harmonious Hodge decomposition.
\begin{Thm}
Let $(A,\LL=D(d+l^\ast_q +\hbar \omega^{-1}))$ be a loop homotopy algebra. For a given harmonious Hodge decomposition $h$, with corresponding inclusion map $i:H(A,d)\to A$ and projection map $\pi:A\to H(A,d)$, there is an associated loop homotopy algebra on the cohomology $H(A,d)$ 
\be\label{eq:deccoder2}
\tilde{\LL}=D(\overset{(\pi)}{\mathsf{T}}_q\circ E(\hbar g^{-1})\circ I +\hbar \tilde{\omega}^{-1}) \;\text{,}
\ee 
where $\tilde{\omega}^{-1}=\pi^{\w2}(\omega^{-1})$ is the projection of $\omega^{-1}$ to the cohomology $H(A,d)$, $\overset{(\pi)}{\mathsf{T}}_q\circ E(\hbar g^{-1})$ represents the graphs with a single output labeled by $\pi$ and $I=e^i$ is the lift of the inclusion map. Furthermore $\tilde{\FF}=\bar{\FF}\circ I$ defines an $IBL_\infty$-isomorphism from $(H(A,d),\tilde{\LL})$ to $(A,\LL)$. 
\end{Thm}

\begin{proof}
From the decomposition theorem we know 
\be\no
\LL\circ\bar{\FF}=\bar{\FF}\circ \bar{\LL} \mand \bar{\LL}^2=0 \;\text{.}
\ee
Furthermore the loop homotopy algebra of the decomposition model is related to the loop homotopy algebra of the minimal model by
\be
\tilde{\LL}=\Pi\circ \bar{\LL}\circ I \sep I\circ \tilde{\LL}= \bar{\LL}\circ I \;\text{,}
\ee
where $\Pi=e^\pi$ is the lift of the projection map. Thus we have
\be\no
\tilde{\FF}\circ\tilde{\LL}=\bar{\FF}\circ I\circ\tilde{\LL}=\bar{\FF}\circ\bar{\LL}\circ I= \LL \circ \bar{\FF}\circ I=\LL\circ \tilde{\FF} \;\text{.}
\ee
Let us denote $\mathcal{P}=e^P$. Using $\mathcal{P}\circ \bar{\LL}\circ I =\bar{\LL}\circ I$, we get
\be\no
\tilde{\LL}^2=\Pi\circ \bar{\LL}\circ \mathcal{P} \circ \bar{\LL}\circ I =\Pi\circ \bar{\LL}^2 \circ I=0 \;\text{.}
\ee
\end{proof}

Finally let us discuss the physical relevance of the minimal model. In the following we abbreviate 
\be\no
\tilde{l}^\ast=\overset{(\pi)}{\mathsf{T}}_q\circ E(\hbar g^{-1})\circ I \;\text{,}
\ee
and thus 
\be\no
\tilde{\LL}=D(\tilde{l}^\ast+\hbar \tilde{\omega}^{-1}) \;\text{.}
\ee
As for the initial loop homotopy algebra, the condition $\tilde{\LL}^2=0$ can be recast into two separate equations, one resembling the main identity (\ref{eq:mainid})
\be\label{eq:mainidmin}
\tilde{l}^\ast\circ D(\tilde{l}^\ast)+\tilde{l}^\ast\circ D(\hbar \tilde{\omega}^{-1})=0 \;\text{,}
\ee
and the other expressing cyclicity (\ref{eq:cyclicity}) with respect to $\tilde{\omega}=\omega\circ i^{\w 2}$, that is
\be\no
\tilde{l}^\ast\circ D(p_i)\w p^i=0 \;\text{,}
\ee
where $\{p_i\}$ denotes a basis of $H(A,d)$ and $\{p^i\}$ denotes its dual basis w.r.t. the symplectic structure $\tilde{\omega}$.

Recall that $\tilde{l}^\ast$ represents the collection of graphs, whose tree lines and loop lines are labelled by $h$ and $g^{-1}=h\circ \omega^{-1}$, respectively (see e.g. figure \ref{fig:graph}). Cyclicity tells us that there is actually no distinction between tree lines and loop lines, this separation is indeed a peculiarity of the formalism. The physical meaningful maps are 
\be\no
\tilde{\omega}(\tilde{l}^\ast_n,\cdot): H(A,d)^{\w n+1}\to \C \;\text{.}
\ee
These are the full quantum S-matrix amplitudes, the sum over all possible Feynman graphs (amputated and restricted to the physical states $H(A,d)$) constructed from the vertices $l^\ast_q$ and the propagator $h$. Finally the main identity (\ref{eq:mainidmin}) summarizes the Ward identities \cite{Zwiebach closed, Verlinde} for the S-matrix amplitudes.

\section{Uniqueness of SFT}
In the previous section we saw that the minimal model theorem is directly related to S-matrix amplitudes. In the following we exploit the more general decomposition theorem and explain its relevance in SFT. With the aid of the decomposition theorem we show in the first subsection that  there is a unique SFT on a given background, compatible with the S-matrix of a given world sheet conformal field theory. In the second subsection we then consider the background independent deformation theory of closed string field theory. Concretely we restrict to deformations which preserve the CFT state space but not the BRST charge and the S-matrix.  Such deformations are described by the Chevalley-Eilenberg cohomology. In particular,  we argue that generic deformations of the closed string vertices are trivial.  This is the closed string analogue of the uniqueness result for open string field theory in \cite{Sachs open-closed}.

\subsection{Fixed Background}\label{sec:uni}
For concretness we present the line of reasoning in the context of quantum closed SFT, but the same conclusion will hold for any bosonic SFT, all we need is actually the decomposition theorem and the concept of RG flow. Consider two string field theories on a fixed string background. The SFTs are determined by a choice of string vertices $\mathcal{V}$ at the geometric level. The world sheet conformal field theory then maps the geometric vertices to the algebraic vertices, preserving the BV structure. As pointed out in the previous sections, consistency requires first that the algebraic vertices define some homotopy algebra and second that the corresponding minimal model coincides with the S-matrix amplitudes of perturbative string theory. Denote the two string field theories by $(A,\LL_0)$ and $(A,\LL_1)$, where $A$ is the (restricted) state space of the world sheet conformal field theory and $\LL_i$ is the loop homotopy algebra defining the SFT. That they are constructed on the same background implies that their BRST differentials and their symplectic structures coincide, i.e.
\be\no
\LL_0=D(d+l^\ast_0+\hbar \omega^{-1}) \mand \LL_1=D(d+l^\ast_1+\hbar \omega^{-1}) \;\text{.}
\ee
Now choose a gauge, such that we can define a propagator $h$, and consider the minimal models corresponding to these two SFTs. Since both SFTs are constructed on the same background, their S-matrix amplitudes are identical and hence their minimal models coincide. Recall that the decomposition model is the sum of the linear contractible part (differential) plus the minimal part. Since both SFTs share the same BRST differential, we can finally conclude that their decomposition models coincide as well. Thus we have
\be\no
\bar{\LL}_0=\bar{\LL}_1\defineR\bar{\LL}\;\text{,}
\ee
where $\bar{\LL}_i$ denotes the decomposition model corresponding to $\LL_i$.
In theorem \ref{thm:dec} we proved that the decomposition model is $IBL_\infty$-isomorphic to its corresponding initial loop homotopy algebra, and because an $IBL_\infty$-isomorphism is invertible the first conclusion is that  two SFTs defined on a given background are $IBL_\infty$-isomorphic:
\be\no
(A,\LL_0)\xleftarrow[]{IBL_\infty-\text{isomorphism}} (A,{\LL})\xrightarrow[]{IBL_\infty-\text{isomorphism}} (A,\LL_1)\;\text{.}
\ee
This is precisely the argumentation of \cite{Kajiura open1}, that was used to show that classical open SFT on a fixed background is unique up to $A_\infty$-isomorphisms. But we can go one step further by tracking the RG flows of the theories: Introduce a UV cut-off $\xi$ for the propagator. The vertices of the action change upon varying the cut-off $\xi$. Geometrically, the variation of the vertices induced by the cut-off scale can be described as follows. To the initial vertices $\mathcal{V}$ we have to attach stubs of length $\xi$. Consistency requires that the string vertices generate a single cover of the full moduli space via Feynman graphs, where the propagator is the operation of sewing in stubs (cylinders) of arbitrary length. Upon attaching stubs to the initial vertices $\mathcal{V}$, we have to add those surfaces that can no longer be produced via Feynman graphs. These surfaces are exactly the ones that arise from graphs where we sew in stubs of length shorter than $2\xi$ \cite{Zwiebach closed, Nakatsu RGopen, Kajiura open2}. Thus for every value of $\xi$ we get a new set of vertices $\mathcal{V}_\xi$. At the algebraic level, the appropriate tool to describe the RG flow is the decomposition model for a certain choice of pre Hodge decomposition. For definiteness let us work in Siegel gauge, where the propagator takes the form 
\be\no
h=-b_0^{+}\int_0^\infty d\tau e^{-\tau L_0^{+}} (1-P) \;\text{,}
\ee
and $P$ is the projection onto physical states, i.e. states annihilated by $L_0^{+}$. Using basic properties of the BRST charge $Q=d$, the energy momentum tensor and the b ghost, we find
\be\no
dh+hd=P-1 \;\text{.}
\ee 
Now the operation of sewing in stubs of length shorter than $2\xi$ corresponds to the map
\be\no
h_\xi=-b_0^{+}\int_0^{2\xi} d\tau e^{-\tau L_0^{+}} \;\text{.}
\ee
Furthermore we have
\be\no
d\,h_\xi+h_\xi \,d=e^{-2\xi L_0^{+}}-1\;\text{,}
\ee
that is we can identify 
\be\no
 P_\xi = e^{-2\xi L_0^{+}} \;\text{.}
 \ee
The map $P_\xi$ is the operation of attaching stubs of length $2\xi$. Recall that the vertices of the decomposition model describe the collection of all graphs with internal lines labelled by the chosen pre Hodge decomposition $h$, the outputs labelled by the corresponding map $P$ and the inputs labelled by the identity map $1$. That is we attach stubs of length $2\xi$ to the outputs and no stubs to the inputs, which is equivalent to attaching stubs of length $\xi$ to outputs and inputs (the two descriptions are $IBL_\infty$-isomorphic).  From the discussion above we can then conclude that the new vertices corresponding to a specific value of the cut-off $\xi$ are given by the decomposition model with the choice of pre Hodge decomposition being $h_\xi$ (see \cite{Kajiura open2} for this discussion in the context of classical open SFT). 

The limit $\xi\to \infty$ describes the decomposition model that corresponds to the S-matrix amplitudes whereas in the limit $\xi \to 0$ we recover the initial loop homotopy algebras. In other words we can interpolate continuously between the S-matrix theory and the initial SFT and thus two SFTs constructed on the same background are connected by a 1-parameter family of $IBL_\infty$-isomorphisms. More precisely we have $IBL_\infty$-isomorphisms parametrized by $t\in[0,1]$
\be\no
\FF(t):(A,\LL(t))\to (A,\LL_0)\;\text{,}
\ee
where $\FF(0)=\op{id}$, $\LL(0)=\LL_0$ and $\LL(1)=\LL_1$.

Recall from appendix \ref{app1} that an $IBL_\infty$-algebra on $A$ is defined to be a Maurer-Cartan element of the Lie algebra $(\coder(SA,\hbar),[\cdot,\cdot])$. The statement that the two loop homotopy algebras $(A,\LL_0)$ and $(A,\LL_1)$ are connected by a 1-parameter family of $IBL_\infty$-isomorphisms implies that they are gauge equivalent Maurer-Cartan elements of $(\coder(SA,\hbar),[\cdot,\cdot])$:
The notion of gauge transformations of $L_\infty$-algebras, and in particular Lie algebras, is reviewed in appendix \ref{app1}. Two Maurer-Cartan elements $\LL_0,\LL_1\in\coder(SA,\hbar)$ are gauge equivalent, if there is a $\Lambda(t)\in \coder(SA,\hbar)$ of degree zero and a $\LL(t)\in \mathcal{MC}\bracketi{\coder(SA,\hbar),[\cdot,\cdot]}$, $t\in[0,1]$, such that
\be\no
\frac{d}{dt}\LL(t)=-[\Lambda(t),\LL(t)] \mand \LL(0)=\LL_0,\quad\LL(1)=\LL_1\;\text{.}
\ee
In our case we have a family of $IBL_\infty$-isomorphisms $\FF(t)$, that is
\be\no
\FF(t)\circ\LL(t)=\LL_0\circ\FF(t) \;\text{,}
\ee
and hence
\begin{align}\no
\frac{d}{dt}\LL(t)&=\frac{d}{dt}\bracketi{\FF(t)^{-1}\circ\LL_0\circ\FF(t)}\\\no
&=\tfrac{d}{dt}\FF(t)^{-1}\circ\FF(t)\circ\LL(t)+\LL(t)\circ\FF(t)^{-1}\circ \tfrac{d}{dt}\FF(t)\\\no
&=-[\Lambda(t),\LL(t)] \;\text{,}
\end{align}
where $\Lambda(t)= \FF(t)^{-1}\circ \frac{d}{dt}\FF(t)$. Thus, we showed that closed SFT on a given background defines a loop homotopy algebra on the (restricted) state space of the world sheet CFT  which is unique up to gauge transformations, or in other words it defines a unique element in the moduli space $\mathcal{M}\bracketi{\coder(SA,\hbar),[\cdot,\cdot]}$. 

\subsection{Uniqueness of closed string field theory}\label{sec:uniclosed}
In \cite{Sachs open-closed} it was shown that a closed string background defines a unique equivalence class of  
classically consistent open string field theories. The equivalence classes are defined w.r.t.  to $L_\infty$ gauge transformations.  In this subsection we will describe the corresponding result for closed string field theory.  For this we first need to  understand the nature of generic gauge transformations and the geometry of the moduli space $\mathcal{M}\bracketi{\op{Coder}(SA),[\cdot,\cdot]}$. Clearly, $L_\infty$ field redefinitions preserve the $L_\infty$ structure and can be interpreted as gauge transformations if they are continuously connected to the identity. On the other hand, field redefinitions include shifts in the closed string background. These are easily seen to be $L_\infty$-isomorphisms along the lines explained in appendix \ref{app1} for $A_\infty$-algebras. This takes us right to he heart of the question about background independence in SFT: For a given homotopy algebra we can consider a non-vanishing Maurer-Cartan element. We then obtain a new homotopy algebra upon conjugation by the Maurer-Cartan element. Background independence then would imply that the structure maps of the minimal model  obtained from this homotopy algebra are equivalent to the perturbative S-matrix elements of the world-sheet CFT in the new background (see figure \ref{fig:back}).

\begin{figure}
\begin{tikzpicture}[node distance=5cm, auto]
  \node (S1) {$S_{CFT}$};
  \node (S2) [below of=S1] {$S^\prime_{CFT}$};
  \node (L1) [right of=S1] {$\{l_n\}_{n\in \N}$};
  \node (L2) [right of=S2] {$\{l_n^\prime\}_{n \in \N}$};
  \draw[->] (S1) to node [swap] {$S_{CFT}\to S_{CFT}+\int \phi$} (S2);
  \draw[->] (S1) to node {operator formalism} (L1);
  \draw[->] (S2) to node {operator formalism} (L2);
  \draw[->] (L1) to node {$L\to E(-c)\circ L \circ E(c)$} (L2);
  \path (S1) to node [swap] {$\mathcal{V}_n\mapsto l_n$} (L1);
  \path (S2) to node [swap] {$\mathcal{V}_n\mapsto l^\prime_n$} (L2);
  \end{tikzpicture}
\caption{background independence}
\label{fig:back}
\end{figure}

In \cite{Kiermaier:2007ba,Kiermaier:2007vu} it was shown that exactly marginal deformation of the open string world-sheet theory correspond to exact solutions, that is Maurer-Cartan elements in open string field theory, thus establishing background independence in one direction, at least in a open neighborhood of a given open string background (see also \cite{Kiermaier:2007vu} for some progress involving marginal deformations). However, since generic string backgrounds are not related by exactly marginal deformations the proof of background independence in general is still not complete. In the next section we will provide an argument for background independence which is sufficient for our purpose. For now we assume background independence. We explained above that background shifts are $\bracketi{\op{Coder}(SA),[\cdot,\cdot]}$-gauge transformations.  
We thus conclude that the perturbative string theories constructed via the world sheet conformal field theories on different closed string backgrounds are within the same  equivalence class. We should note, however, that field redefinitions do not preserve the decomposition of the homotopy algebra and, in particular, background shifts do not preserve the cohomology $H(A,Q)$.  We now want to argue that all infinitesimal deformations of closed string field theory are trivial. 

The proof of this assertion proceeds in close analogy with the corresponding open string result (section 4.2 of \cite{Sachs open-closed}). Let us denote the classical closed string vertices by $f_n\equiv f^{0,0}_{n,0}$ (see section \ref{sec:infty}). The bracket $ [\cdot,\cdot]$ on $\op{Coder}(SA)$ induces the Chevalley-Eilenberg differential $d_C=[L,\cdot]$ on the deformation complex. Any consistent infinitesimal deformation $\Delta f=\{\Delta f_n\}_{n\in \N}$ of the $L_\infty$-structure $\{f_n\}_{n\in \N}$  is  $d_C$-closed, $d_C( \Delta f)=0$. Starting with $n=2$ we conclude $(l_1\Delta f_2)( c_1,c_2)\equiv \Delta f_2(l_1 c_1,c_2)+(-1)^{c_1}\Delta f_2( c_1,l_1c_2)=0$ which implies that $\Delta f_2(c_1,c_2)=\omega_c({\Delta l_1}c_1, c_2)$ with $[l_1, {\Delta l_1}]=0$. For $n=3$ we write 
\begin{equation}
\Delta f _3(c_1,c_2,c_3)=\omega_c (\Delta l_2(c_1,c_2),c_3) 
\end{equation}
Then $\Delta f _2$ and $\Delta f _3$ are subject to the equation
\begin{equation}\label{cc2}
(l_2 \Delta f_2)+(l_1 \Delta f_3)=0.
\end{equation}
It is not hard to see that this implies that $\Delta l_2=[{\mathcal O},l_2]+g_2$ with ${\mathcal O}$ a linear operator and  $[l_1,g_2]=0$. To continue we can assume without restricting the generality that bpz$( {\mathcal{O}})=\pm {\mathcal O}$. If ${\mathcal O}$ is BPZ-odd then 
\begin{equation}
\Delta f_2(c_1,c_2)=\omega_c( [{\mathcal O},l_1] c_1,c_2) +\omega_c( {\mathcal H}c_1, c_2)
\end{equation}
where  $[l_2,{\mathcal H}]=0$. The latter condition together with $[l_1,{\mathcal H}]=0$ from above is in contraction with the uniqueness of the world-sheet BRST charge $Q$. Thus ${\mathcal H}=0$. 
Furthermore, for ${\mathcal O}$ BPZ-odd, $\Delta f_2$ and $\Delta f_3$ are exact. 
This leaves us with ${\Delta l_1}=0$ and bpz$( {\mathcal O})= {\mathcal O}$. 
If ${\mathcal O}$ and $g_2$ are $l_1$-exact then $\Delta f_3$ is again trivial in the $d_C$-cohomology.
To continue, we then assume that ${\mathcal O}$ and $g_2$ are in the cohomology of
$l_1$ and consider $n=4$ which gives the condition
\begin{equation}\label{cc3}
(l_3 \Delta f_2)+ (l_2 \Delta f_3)+(l_1 \Delta f_4)=0
\end{equation}
However, $\Delta f_2=0$ from the above and 
\begin{equation}\label{dc3}
(l_2 \Delta f_3)=\omega_c([\Delta l_2,l_2](c_1,c_2,c_3),c_4) 
\end{equation}
Now, since ${\mathcal O}$ and $g_2$ are in the cohomology of $l_1$ this term cannot be canceled by $(l_1\Delta f_4)$ unless $g_2=0$ and ${\mathcal O}$
is a conformal invariant so that ${\mathcal O}$ can be pulled in the bulk.  Indeed, since ${\mathcal O}$ is not $l_1$-exact, the only way the
differential $l_1$ acting on $\Delta f_4$ can reproduce (\ref{dc3}) is as a
derivation on its moduli space. Since ${\mathcal O}$ is BPZ-even it
cannot be a derivative. On the other hand if ${\mathcal O}$ can be pulled
in the bulk then $[{\mathcal O},l_1]= 0$ is equivalent
to the closed string cohomology condition. Repeating these steps for $n>4$ one obtains the desired result.

It then follows that the only non-trivial elements in
the $d_C$-cohomology are given by a degree 0 insertion of
the form ${\mathcal{O}}|0\rangle_c$. However,  since the semi-relative closed string cohomology at degree $-2$ (ghost number zero)\footnote{The ghost number is related to the grading use here by a double shift (see e.g. \cite{qocha}), i.e. the ground state $|0\rangle_c$ has degree minus two and ghost number zero, respectively.} and vanishing mass contains only the vacuum (or
equivalently the identity operator ${\mathcal{O}}=I$), we conclude that there are no non-trivial deformations of closed string field theory.

\section{Open-closed Correspondence}\label{sec:oc}
Let us turn to the theory of open and closed string. In the geometrical setting of bounded Riemann surfaces, it is generically impossible to distinguish whether a surface should be interpreted as the world sheet of a propagating open or closed string. From the point of view of open strings, a cylinder for example represents a one-loop diagram, whereas the alternative identification is the closed string propagator. There is an algebraic counterpart to this phenomenon which we will investigate. The main result of this section is then to describe an isomorphism between deformations of open string theory and closed string Maurer-Cartan elements. 

\subsection{Open-closed correspondence}
Consider open-closed SFT in the 'classical' limit as described in section \ref{sec:infty}. That is we have vertices corresponding to discs with open, discs with open and closed and spheres with closed string punctures. The open-closed vertices define a $L_\infty$-morphism from the $L_\infty$-algebra of closed strings to the differential graded Lie algebra which controls deformations of the open string $A_\infty$-algebra (see equation (\ref{eq:ocha})). The OCHA (\ref{eq:ochaex}) of \cite{Kajiura open-closed1,Kajiura open-closed2} reads
\be\no
N\circ L = d_h(N)+\inv{2}[N,N]\circ \Delta \;\text{,}
\ee
where $N$ represents the open-closed vertices, $L$ represents the closed vertices and $d_h=[M,\cdot]$ with $M$ representing the open string vertices. $L_\infty$-morphisms preserve Maurer-Cartan elements, thus let us identify the Maurer-Cartan elements on the closed and open side of the OCHA. The Maurer-Cartan elements of the closed string $L_\infty$-algebra are solution of the equations of motion, whereas on the open string side a Maurer-Cartan element of the differential graded Lie algebra $(\op{Coder}^{cycl}(TA),d_h,[\cdot,\cdot])$ defines a finite deformation of the $A_\infty$-algebra $M$. Thus every solution of the closed string equations of motion defines a new open string field theory. This is the classical open-closed correspondence \cite{Kajiura open-closed1,Kajiura open-closed2}. At the infinitesimal level, the open-closed vertex with just one closed input $N_1\defineL N \circ i_1$ defines a morphism of complexes, that is it maps physical closed string states to infinitesimal deformations of the open string field theory. Indeed we know more about this vertex. In \cite{Sachs open-closed} it has been shown that $N_1$ defines a quasi-isomorphism, that is it induces an isomorphism on cohomologies. A powerful theorem of Kontsevich \cite{Kontsevich} then guarantees isomorphism at the finite level, or more precisely that the moduli spaces of two $L_\infty$-algebras 
connected by a $L_\infty$-quasi-isomorphism are isomorphic. In our particular case, this means that the space of closed stings satisfying the equations of motion modulo gauge transformations is in one-to-one correspondence with the space of inequivalent deformations of the open string field theory $M$, i.e. 
\be\no
\mathcal{M}(A_c,L)\cong \mathcal{M}(\op{Coder}^{cycl}(TA_o),d_h,[\cdot,\cdot]) \;\text{.}
\ee

In the previous section we argued that SFT is unique (up to gauge transformations) on a given background, thus we need not distinguish between SFT and world sheet conformal field theory. Let us then formulate the open-closed correspondence in terms of world sheet conformal field theories. We start with an open-closed world sheet conformal field theory. The restriction to open/closed strings induces an open/closed world sheet conformal field theory. The moduli space of the $L_\infty$-algebra corresponding to the closed world sheet conformal field theory is isomorphic to the space of inequivalent open world sheet conformal field theories.

Since the open-closed correspondence relates Maurer-Cartan elements modulo gauge transformations, we give some examples of gauge transformations in order to develop some intuition. On the closed side   we know what gauge transformations are, they leave the equations of motion invariant. Thus we will focus on the open side where we are in the context of $A_\infty$-algebras, and represent three types of gauge transformations therein. In the following $(A_o,M)$ is the $A_\infty$-algebra describing the open SFT.
\bi
\item[\emph{1.}]\emph{1-parameter family of $A_\infty$-isomorphisms:}\hspace{2mm} 
For $t\in[0,1]$, let $(A_o,M_t)$ be $A_\infty$-algebras connected continuously to the initial $A_\infty$-algebra $(A_o,M=M_{t=0})$ by $A_\infty$-isomorphisms 
\be\no
F_t:(A_o,M_t)\to(A_o,M)\;\text{.}
\ee
In general, a gauge transformation of the differential graded Lie algebra $(\op{Coder}^{cycl}(TA_o),[\cdot.\cdot])$ is given by a a 1-parameter family of $A_\infty$-algebras $M_t\in\op{Coder}^{cycl}(TA_o)$, satisfying
\be\label{eq:gaugetransopen}
\frac{d}{dt}M_t=-[\Lambda_t,M_t] \;\text{,}
\ee
for some $\Lambda_t\in \op{Coder}^{cycl}(TA_o)$ of degree zero. Thus we conclude that $(A,M)$ is gauge equivalent to $(A,M_1)$, with $\Lambda_t=F^{-1}_t\circ\tfrac{d}{dt} F_t$.
\item[\emph{2.}]\emph{backgrounnd shifts:}\hspace{2mm} 
A more concrete example is that of shifting the open string background. Let $a\in A_o$ be an open string state of degree zero. A background shift in the initial $A_\infty$-algebra gives rise to a new $A_\infty$-algebra $(A,M[a])$, defined by (see appendix \ref{app1})
\be\no
M[a]=E(-a)\circ M\circ E(a) \;\text{,}
\ee 
where $E(a)=e^{1+a}$ is the lift of the identity map $1$ plus the background $a$ to a cohomomorphism. $E(a)$ defines in fact an $A_\infty$-isomorphism from $(A,M[a])$ to $(A,M)$. We can easily construct a 1-parameter family of $A_\infty$-isomorphisms by gradually scaling the background to zero, that is the $A_\infty$-isomorphisms are $E(ta)$ and $\Lambda_t=\liftd{a}$, where $\liftd{a}$ denotes the lift of $a$ to a coderivation.
\item[\emph{3.}]\emph{attaching strips - RG flow:}\hspace{2mm}
In section \ref{sec:uni} we used the decomposition theorem to discuss the RG flow in closed SFT. Introducing a cut-off $\xi$ for the propagator amounts to attaching stubs of length $\xi$ to the vertices. 
In the case of open SFT and $A_\infty$-algebras a similar discussion can be found in \cite{Kajiura open2}: Attach strips of length $\xi$ to the initial vertices and add those diagrams to the vertices that can no longer be produced by (tree-level) Feynman graphs. The new vertices $M_\xi$ are given by the decomposition model for a suitable choice of pre Hodge decomposition, and the decomposition theorem provides an $A_\infty$-isomorphism $F_\xi:(A,M_\xi)\to(A,M)$. Since we can vary the length of the stubs continuously, the initial SFT and the one with strips attached are related by a 1-parameter family of $A_\infty$-isomorphisms and are thus gauge equivalent.
\ei

Let us now return to the question of background independence.  We explained above that there is an isomorphism between classical solutions of closed string field theory, modulo closed string gauge transformations, and the moduli space $\mathcal{M}(\op{Coder}^{cycl}(TA_o),d_h,[\cdot,\cdot])$ of  inequivalent open string $A_\infty$ structures. Now suppose that for a given non-trivial closed string Maurer-Cartan element $ c$, $L(e^{ c})=0$, we obtain the corresponding open string theory via the operator formalism of the world-sheet CFT on the background  $\phi$ (see figure \ref{fig:back}). The corresponding $A_\infty$ structure is then necessarily given by an element  $M_\phi \in \mathcal{M}(\op{Coder}^{cycl}(TA_o),d_h,[\cdot,\cdot])$. Since this space is isomorphic to the moduli space of  classical solutions of closed string field theory ($L_\infty$ Maurer-Cartan elements) we can identify $M_\phi$ with the image of $e^{ c}$ under the open-closed $L_\infty$-morphism
 \begin{eqnarray}
N: SA_c &\to& \op{Coder}(TA_o)\nonumber\\
e^{ c}&\mapsto& N(e^{ c})=M_\phi
\end{eqnarray}
In addition to the open string vertices on the background $\phi$, the operator formalism  defines a (closed string) $L_\infty$-algebra $L_\phi$, as well as an $L_\infty$-morphism $N_\phi\;:\; (A_c,L_\phi)\to(\op{Coder}(TA_o),[M_\phi,\cdot],[\cdot,\cdot])$.
Since the vector space $A_c$ is assumed to be invariant $L_\phi$ is related to $L$ by an invertible $L_\infty$-field redefinition $K: SA_c\to SA_c$ so that $N_\phi=N\circ K$. On the other hand, we have from the above that 
\be 
N_\phi(1)=N(e^{ c})=(N\circ E( c))(1) \;\text{.}
\ee
 From this we then conclude that $K=E( c)$ and thus $L_\phi=E(- c)\circ L\circ E( c)$. This then proves independence of backgrounds which leave the vector space $A_c$ (state space of the world-sheet CFT) invariant. 
 
\subsection{Quantum case}
In the previous subsection we discussed the open-closed correspondence as it arises from 'classical' open-closed SFT. The correspondence is based on the property that $L_\infty$-morphisms preserve Maurer-Cartan elements and the way the OCHA (\ref{eq:ocha}) is defined. The algebraic structure of quantum open-closed SFT is quite similar to that of the OCHA: The open-closed vertices define an $IBL_\infty$-morphism from the loop algebra of closed strings to the involutive Lie bialgebra on the cyclic Hochschild complex of open strings (\ref{eq:qocha1}). As in the classical case, $IBL_\infty$-morphisms preserve Maurer-Cartan elements and thus we can look for closed string Maurer-Cartan elements which will in turn define consistent quantum SFTs of only open strings. The $L_\infty$-morphism in the classical case was shown to be a $L_\infty$-quasi-isomorphism \cite{Sachs open-closed}, which implies that the $IBL_\infty$-morphism is a quasi-isomorphism as well. Furthermore the moduli spaces of $IBL_\infty$-quasi-isomorphic $IBL_\infty$-algebras are isomorphic \cite{Cieliebak ibl}. Thus quantum open SFTs, if there exists any, are in one-to-one correspondence with Maurer-Cartan elements of the closed string loop algebra (up to gauge transformations). Therefore let us investigate the Maurer-Cartan equation of the closed string loop algebra 
\be\no
\LL_c=D(d+l_q^\ast+\hbar \omega_c^{-1})\in \coder(SA_c,\hbar) \;\text{,}
\ee
as described in section \ref{sec:dec}. A Maurer-Cartan element $\cc=\sum_{n,g}\hbar^{n+g-1}c^{n,g}$, $c^{n,g}\in A_c^{\w n}$, of $(A_c,\LL_c)$ satisfies
\be\label{eq:MCloop}
\LL_c(e^{\cc})=0 \;\text{.}
\ee
The corresponding quantum open string field theory $\mm[\cc]$ is defined by
\be\no
\lifte{\mm[\cc]}=\lifte{\nn}(\lifte{\cc})\;\text{,}
\ee
and satisfies
\be\label{eq:qAinfty}
\LL_o(\lifte{\mm[\cc]})=\lifte{\nn}\circ \LL_c (\lifte{\cc})=0\;\text{,}
\ee
due to equation (\ref{eq:qocha1}). Similarly as in equation (\ref{eq:qocha2}), equation (\ref{eq:qAinfty}) can be recast into
\be\no
\liftd{d}_h(\mm[\cc])+\liftd{[\cdot,\cdot]}(\mm[\cc])+\inv{2}\liftd{[\cdot,\cdot]}(\mm[\cc]\w\mm[\cc])-\liftd{[\cdot,\cdot]}(\mm[\cc])\w\mm[\cc]=0\;\text{,}
\ee
which is the defining equation of a quantum $A_\infty$-algebra \cite{Herbst qAinfty}. The closed string Maurer-Cartan equation (\ref{eq:MCloop}) was analyzed in \cite{qocha} and implies the following:
\bi
\item[(i)]
$c\defineL c^{1,0}$ has to satisfy the classical Maurer-Cartan equation 
\be\no
\sum_{n=0}^\infty \inv{n!} l^0_n(c^{\w n})=0 \;\text{,}
\ee
that is $c$ defines a closed string background.
\item[(ii)]
Consider the part $g^{-1}\defineL c^{2,0}$ of the Maurer-Cartan element $\cc$. Contracting one output of $g^{-1}$ with the symplectic structure $\omega_c$ we obtain a linear map
\be\no
h\defineL g^{-1}\circ\omega_c : A_c \to A_c \;\text{.}
\ee
In leading order in $\hbar$, we found in \cite{qocha} that the part of the Maurer-Cartan equation with two outputs implies
\be\label{eq:BRSTtrivial}
d[c]\circ h + h \circ d[c]=-1 \;\text{.}
\ee 
Equation (\ref{eq:BRSTtrivial}) states that the cohomology of $d[c]$ is trivial, i.e. that there are no physical states in the background $c$. Furthermore we can identify $h$ as the propagator corresponding to this background and $g^{-1}$ as the inverse metric on the unphysical states (see section \ref{sec:dec}). Note that in order to derive equation (\ref{eq:BRSTtrivial}) we had to contract with the symplectic structure, and the identity map on the right hand side of equation (\ref{eq:BRSTtrivial}) stems from the assumption that $\omega_c$ is non-degenerate. This observation will be crucial in the topological string, since there the symplectic structure degenerates on the physical states and BRST triviality does not follow from the Maurer-Cartan equation in that case.
\ei
From these two observations we will now prove by contradiction that the loop homotopy algebra of closed strings does not admit any Maurer-Cartan element (assuming that $\omega_c$ is non-degenerate). 
Assume that $\cc$ is a Maurer-Cartan element of $\LL_c$ and consider the background shifted loop algebra
\be\no
\LL_c[c]=E(-c)\circ \LL_c \circ E(c)=D(d[c]+l_q^\ast[c]+\hbar \omega_c^{-1})\;\text{.}
\ee
Again $c=c^{1,0}$ and $g^{-1}=c^{2,0}$. Next we construct the minimal model $(H(A_c,d[c]),\tilde{\LL}_c[c])$ of $(A_c,\LL_c[c])$. Since the cohomology of $d[c]$ is trivial, i.e. $H(A_c,d[c])=\{0\}$, the only candidate Maurer-Cartan element of $(H(A_c,d[c]),\tilde{\LL}_c[c])$ is $0$, but $\tilde{\LL}_c[c](\lifte{0})=\tilde{\omega}_c^{-1}\neq 0$. Thus $(H(A_c,d[c]),\tilde{\LL}_c[c])$ has no Maurer-Cartan elements and likewise
$\mathcal{M}(H(A_c,d[c]),\tilde{\LL}_c[c])=\emptyset$. Since by construction, $(H(A_c,d[c]),\tilde{\LL}_c[c])$ is quasi-isomorphic to $(A_c,\LL_c[c])$ and quasi-isomorphic $IBL_\infty$-algebras have isomorphic moduli spaces \cite{Cieliebak ibl}, we conclude
\be\label{eq:moduli}
\emptyset=\mathcal{M}(H(A_c,d[c]),\tilde{\LL}_c[c])\cong \mathcal{M}(A_c,{\LL}_c[c])\cong  \mathcal{M}(A_c,{\LL}_c)\;\text{.}
\ee
The second isomorphism in (\ref{eq:moduli}) follows from the observation that a Maurer-Cartan element $\cc$ of $\LL_c[c]$ corresponds to a Maurer-Cartan element $\cc+c$ of $\LL_c$.

Equation (\ref{eq:moduli}) states that there are no Maurer-Cartan elements of the closed string loop algebra, which in turn implies that there is no consistent quantum theory of only open strings, or in other words it is impossible to deform the classical open string field theory determined by an $A_\infty$-algebra $m$ into a quantum $A_\infty$-algebra $\mm$.

\subsection{Illustration}
We will illustrate this last point with an simple toy example: We consider a differential Lie algebra $(A,d,[\cdot,\cdot])$. A harmonious Hodge decomposition is then defined by the triple $d,h$ and\footnote{For $A=\Omega^\bullet(M)$, the space of differential forms on $M$ we have $h=-\frac{d^\dagger}{\Delta}$.}  $P$. The classical Maurer-Cartan equation 
\be\label{MCI1}
dc+[c,c]=0
\ee
implies that the $L_\infty$ structure of the corresponding minimal model follows from the equation $P[c,c]=0$, where $c=c_P+c_U+c_T$ is recursively determined through 
\be
c_U=h [c,c] \;\text{.}
\ee
Thus, $\tilde l_2(c_P,c_P)=P[c_p,c_p]$, $\tilde l_3(c_P,c_P)=P[c_P,h[c_P,c_P]]+\cdots$ and so forth. 

Let us now turn to the quantum Maurer-Cartan equation 
\be 
D(d+[\cdot,\cdot]+\hbar\omega^{-1})(e^{c+\hbar g^{-1}})=0
\ee
where $g^{-1}=\frac{1}{2}u_i\wedge {}^iu\in A\wedge A$. We can disentangle this equation by successive projections as in section VI.A of \cite{qocha} onto $A$, $A\wedge  A$ and $A\wedge A\wedge  A$ respectively. This gives 
\begin{eqnarray}
0&=& dc+[c,c]+\frac{\hbar}{2} [u_i, {}^iu]\label{MCIa}\\
0&=&d u_i\wedge {}^iu+[c, u_i]\wedge {}^iu+\omega^{-1}\label{MCIb}\\
0&=& [u_i, u_j]\wedge {}^iu\wedge {}^ju\label{MCIc}
\end{eqnarray}
where (\ref{MCIb}) and  (\ref{MCIc}) come with a global factor of $\hbar$ and $\hbar^2$ respectively. At order $\hbar^0$ we recover the classical Maurer-Cartan equation (\ref{MCI1}). The obstructions at the quantum level arise from (\ref{MCIb}). Upon composing  (\ref{MCIb}) with $\omega$ to the right we recover the propagator equation in the background $c$, 
\be 
d[c]\circ h+h\circ d[c] = P[c]-1\;\text{,}
\ee
provided $\omega$ is degenerate, i.e. vanishes on $H(A)$. If $\omega$ is non-degenerate then (\ref{MCIb}) has no solutions and consequently the quantum moduli space is the empty set. This is the case in bosonic string theory. Further obstructions can arise from (\ref{MCIc}). In Chern-Simons theory (\ref{MCIc}) is compatible with the propagator equation but we cannot exclude obstructions arising from (\ref{MCIc}) in general. The topological string, discussed in the next section, is another example where $\omega$ is degenerate and (\ref{MCIb}) and (\ref{MCIc}) are compatible.

\section{Applications to topological strings}\label{sec:top}
The world sheet description of the topological string is based on a supersymmetric sigma model whose target space is a Calabi-Yau 3-fold, X (see \cite{Witten:1991zz} for a good review).  The world sheet theory admits 4 supercharges as well as an R-symmetry current and the energy momentum tensor. The latter can be twisted by the R-current in such a way that a linear combination of the supercharges defines a differential $Q$ on the state space. Furthermore, the world sheet theory is topological on the cohomology of $Q$.  There are two possible ways to twist the energy momentum tensor leading to two inequivalent theories, the A-model and the B-model. The algebra of the triplet consisting of the differential $Q$, together with the remaining supercharge and the stress tensor is isomorphic to that of $Q$, $b$-ghost and the energy momentum tensor of the BRST quantized bosonic string CFT. 
Thus we can apply the operator formalism as in bosonic string field theory to define vertices. The corresponding field theories have been constructed in \cite{Witten CS} in the open case, and in the closed case in \cite{Vafa KS} and \cite{Bershadsky KG} for the B- and A-model, respectively. 

In the A-model there is a natural chain map between the de Rham complex of $X$ and the BRST complex of the twisted world sheet sigma model.  If one restricts to local operators this map induces an isomorphism between the de Rham cohomology  and the BRST cohomology. In particular, the degree $(1,1)$ elements of the BRST cohomology of the twisted world sheet sigma model are identified with the K\"ahler structure of $X$. 

In the B-model, on the other hand, there is a chain map between the BRST cohomology and 
$\oplus_{p.q}H^p(X,\wedge^qT^{1,0} X)$. Again, this induces an isomorphism on the cohomology upon restriction to local operators. Consequently the degree $(1,1)$ elements of the BRST cohomology are identified with the changes of complex structure of $X$. 

Although bosonic and topological string theory share some fundamental properties, there are many crucial differences which we summarize here:
\bi
\item[(i)] The action of topological open/closed string theory is cubic \cite{Witten CS}/\cite{Vafa KS,Bershadsky KG}. Furthermore these actions satisfy the quantum BV master equation.
In particular, the closed string vertices define a loop homotopy algebra without including higher vertices.
\item[(ii)] The operator $b_0^{-}$ generically does not have trivial cohomology. Thus it is impossible to define an operator 
$c_0^{-}$, such that $\{b_0^{-},c_0^{-}\}=1$. Such an operator exists only on all but the physical states, where the physical states are identified with the kernel of $L_0^+$ (i.e. we work in Siegel gauge).
In other words, there is an operator $c_0^{-}$, such that
\be\no
\{b_0^{-},c_0^{-}\}=1-P \;\text{,}
\ee
where $P$ is the projection onto the physical states \cite{Vafa KS, Bershadsky KG}.
\item[(iii)] The minimal models corresponding to the off-shell $L_\infty$-algebras of closed strings vanish identically \cite{Kajiura open-closed2}. 
\item[(iv)]Following the prescription of \cite{Zwiebach closed}, the symplectic structure on the closed string side is defined by inserting the operator $c_0^{-}$ into the inner product (bpz inner product in the context of bosonic string theory), that is
\be\no
\omega_c=(\cdot,c_0^{-}\cdot)\;\text{.}
\ee
Since $c_0^{-}$ is defined only on the trivial and unphysical states, the symplectic structure of closed strings degenerates on the physical states. Thus we conclude that the entire minimal model of the loop homotopy algebra of closed strings vanishes - the vertices and the symplectic structure. 
\ei

Let us now turn to the open-closed correspondence in the context of topological string theory. At the classical level, the open-closed vertices again define an $L_\infty$-morphism from the $L_\infty$-algebra defined by the closed string vertices to the Hochschild complex of the open string $A_\infty$-algebra. For the B-model, it has been shown in \cite{Hofman Bmodel} that the $L_\infty$-morphism is indeed a quasi-isomorphism. Furthermore they conjectured that this should be the case in any string field theory realization of the OCHA, which has been confirmed for the bosonic string \cite{Sachs open-closed} and also for Landau-Ginzburg models \cite{Carqueville LG}. At a more abstract level, the results of \cite{Costello} seem to support this conjecture as well. This implies that the moduli spaces of the $L_\infty$-algebras connected by the quasi-isomorphism are isomorphic, that is
\be\no
\mathcal{M}(A_c,L_c)\cong \mathcal{M}(\op{Coder}^{cycl}(TA_o),d_h,[\cdot,\cdot]) \;\text{,}
\ee
where $L_c$ denotes the $L_\infty$-algebra of closed strings and $d_h=[M,\cdot]$ is the Hochschild differential corresponding to the open string $A_\infty$-algebra $M$. An $L_\infty$-algebra is quasi-isomorphic to its minimal model which implies isomorphy of their respective moduli spaces. Recall that one of the distinguished properties of topological strings is that the closed string minimal model vanish identically, and thus we conclude
\be\no
\mathcal{M}(\op{Coder}^{cycl}(TA_o),d_h,[\cdot,\cdot])\cong\mathcal{M}(A_c,L_c)\cong\mathcal{M}(H(A_c,d),\tilde{L_c})=H_0(A_c,d_c)\;\text{,}
\ee
where $\tilde{L_c}=0$ denotes the minimal model corresponding to $L_c$ and $H_0(A_c,d_c)$ represents the cohomology of $d$ at degree zero. That is, inequivalent deformations of topological open string field theory are parametrized by physical closed string states. 

On the other hand one can also ask for deformations of topological open string (tree-level) amplitudes induced by closed strings \cite{Carqueville LG}. To attempt this question, it is useful to think of the OCHA as a single algebraic entity and take the minimal model of it \cite{Kajiura open-closed1}. The minimal model of an OCHA is described by the minimal model of its closed string $L_\infty$-algebra linked to the deformation complex of the minimal model of its open string $A_\infty$-algebra by an $L_\infty$-morphism. If the $L_\infty$-morphism of the initial OCHA is a quasi-isomorphism, then so is the $L_\infty$-morphism of the corresponding minimal model. This implies
\be\no
\mathcal{M}(\op{Coder}^{cycl}(TH(A_o,d_o)),\tilde{d}_h,[\cdot,\cdot])\cong\mathcal{M}(H(A_c,d),\tilde{L}_c)=H_0(A_c,d_c)\;\text{,}
\ee
where $\tilde{d}_h=[\tilde{M},\cdot]$ is the Hochschild differential induced by the minimal model of the open string $A_\infty$-algebra and $H(A_o,d_o)$ represents the physical open string states. 
In other words, physical closed string states parametrize the space of inequivalent deformations of topological open string (tree-level) amplitudes\footnote{This space includes deformations with a non-vanishing tadpole. Amplitudes without tadpole correspond to the moduli space of the full OCHA \cite{Kajiura open-closed1, Kajiura open-closed2}.}. 

Now we draw our attention to the quantum case. As state previously, the cubic closed string action satisfies the quantum BV master equation and thus defines a loop homotopy algebra. If the symplectic structure $\omega_c$ is non-degenerate, we showed in section \ref{sec:oc} that the corresponding loop homotopy algebra does not admit any Maurer-Cartan element at all. In the topological string, the symplectic structure degenerates on the physical states. This implies that equation (\ref{eq:BRSTtrivial}) is modified to
\be\no
d_c[c]\circ h + h\circ d_c[c]=P[c]-1 \;\text{,}
\ee
where $P[c]$ is the projection onto the physical states in Siegel gauge. This is the propagator equation or in mathematical terms $h$ is a harmonious Hodge decomposition of $A_c$ (see section \ref{sec:dec}). 
Thus the Maurer-Cartan equation of the loop homotopy algebra of topological strings does not require a vanishing BRST cohomology, and hence, in contrast to bosonic string theory, the conclusion that there cannot be any Maurer-Cartan elements does not persist here. Similar to bosonic string field theory, the full open-closed theory defines a QOCHA, where the open-closed vertices define an $IBL_\infty$-morphism from the loop homotopy algebra of closed strings to the involutive Lie bialgebra on the cyclic Hochschild complex of open strings. The $IBL_\infty$-morphism is a quasi-isomorphism, since the classical $L_\infty$-morphism is, and thus the moduli spaces of the respective $IBL_\infty$-algebras are isomorphic. 
\begin{align}\label{eq:topqoc}
\mathcal{M}(\mathcal{A}_o,\LL_o)&\cong \mathcal{M}(A_c,\LL_c)\\\no
&\cong \mathcal{M}(H(A_c,d_c),\tilde{\LL_c})\\\no
&=\big\{\cc=\sum_{n,g}\hbar^{n+g-1}c^{n,g} \,\big|\, c^{n,g}\in H(A_c,d_c)^{\w n},\; |c^{n,g}|=0\big\}
\end{align}
In equation (\ref{eq:topqoc}) $\mathcal{A}_o=\op{Hom}^{cycl}(TA_o,R)$ denotes the cyclic Hochschild complex, $\LL_o=\liftd{d}_h+\liftd{[\cdot,\cdot]}+\hbar \liftd{\delta}$, $\LL_c$ is the closed string loop homototpy algebra and  $\tilde{\LL}_c=0$ its corresponding minimal model (see appendix \ref{app1} and \ref{sec:dec}). Maurer-Cartan elements of $\LL_o$ represent deformations of the initial $A_\infty$-algebra $M$ to a quantum $A_\infty$-algebra, or in other words, they represent consistent quantum theories of only open strings. Equation (\ref{eq:topqoc}) states that the space of quantum open string theories is parametrized by symmetric tensors in $H(A_c,d_c)$ of degree zero, which generalizes the classical open-closed correspondence where we allowed just for vectors. As in the classical case, we can also ask for bulk induced deformations of open string amplitudes (including loops). Again the idea is to take the minimal model of the whole QOCHA, which is guaranteed to exist due to \cite{Lazarev1,Lazarev2}, and leads to the statement that
\be\no
\mathcal{M}(\tilde{\mathcal{A}}_o,\tilde{\LL}_o)\cong \mathcal{M}(H(A_c,d_c),\tilde{\LL_c}) =\big\{\cc=\sum_{n,g}\hbar^{n+g-1}c^{n,g} \,\big|\, c^{n,g}\in H(A_c,d_c)^{\w n},\; |c^{n,g}|=0\big\} \;\text{,}
\ee  
where $\tilde{\mathcal{A}}_o=\op{Hom}^{cycl}(TH(A_o,d_o))$ and $\tilde{\LL}_o=\liftd{\tilde{d}}_h+\liftd{[\cdot,\cdot]}+\hbar \liftd{\delta}$.
Thus we find that the topological open string amplitudes can be deformed by closed strings in a more general way then discussed in \cite{Herbst qAinfty}. It is not just closed string backgrounds but also higher rank tensors that deform the topological open string amplitudes. In order to get a world-sheet interpretation of such deformations we recall the chain map from the de Rham complex (A-model) to the BRST complex reviewed at the beginning of this section. Correspondingly tensor deformations are implemented on the world-sheet by non-local CFT operators. It would be interesting to see if non-trivial deformations of this type exist.

\section{Outlook}
In this paper we discussed several properties of bosonic string field theory in terms of homotopy algebras. In particular, combining the open-closed homotopy algebra with the isomorphism between consistent infinitesimal deformations of classical open string field theory and physical closed string states, we established an isomorphism between closed string Maurer-Cartan elements and consistent finite deformations of open string field theory. The QOCHA also provides a simple algebraic description of the obstructions (notably absent in the topological string) to the existence of Maurer-Cartan elemnts at the quantum level.  

We also proved a decomposition theorem for the loop algebra of quantum closed string field theory which, in turn, implies uniqueness of closed string field theory on a given background. Finally, we also addressed uniqueness and background independence of closed string field theory using OCHA. 

In contrast, a complete formulation of super string field theory has not been developed yet \cite{Witten:1986qs,Berkovits:2000fe}. Generalizing the prescription of \cite{Zwiebach closed, Zwiebach open-closed} to the supersymmetric case, the first task would be to construct a BV algebra on the singular chains of super Riemann surfaces. The super conformal field theory of the super string is then expected to define a morphism of BV algebras and would lead to some novel algebraic structures on the corresponding state space.

\vspace{2cm}\noindent {\bf Acknowledgements}: The authors would like to thank Barton Zwiebach for many informative discussions and subtle remarks as well as Branislav Jurco, Kai Cieliebak and Sebastian Konopka for helpful discussions. K.M. would like to thank  Martin Markl and Martin Doubek who stimulated his interest in operads and their applications to string field theory. This project was supported in parts by the DFG Transregional Collaborative Research Centre TRR 33,
the DFG cluster of excellence ``Origin and Structure of the Universe'' as well as the DAAD project  54446342,
I. S. would like to thank the Center for the fundamental laws of nature at Harvard University for hospitality during the initial stages of this project.


\newpage
\appendix

\section{$A_\infty$-, $L_\infty$- and $IBL_\infty$-algebras}\label{ssec:ALIBLinfty}\label{app1}

\subsection{$A_\infty$- and $L_\infty$-algebras}
Let $A=\oplus_{n\in\Z}A_n$ be a graded module over some ring $R$ and consider the tensor algebra 
\be\no
TA=\bigoplus_{n=0}^\infty A^{\tp n}\;\text{,}
\ee
with comultiplication $\Delta:TA\to TA\tp TA$ defined by
\be\no
\Delta (a_1\tp\dots\tp a_n)= \sum_{i=0}^n (a_1\tp \dots\tp a_i)\tp(a_{i+1}\tp\dots\tp a_n)\;\text{.}
\ee
We have the canonical projection maps $\pi_n:TA\to A^{\tp n}$ and inclusion maps $i_n:A^{\tp n}\to TA$.

A coderivation $D\in \op{Coder}(TA)$ is a linear map on $TA$ that satisfies
\be\no
(D\tp \op{id}+\op{id}\tp D)\circ \Delta = \Delta\circ D \;\text{.}
\ee
From this property it follows that there is an isomorphism $\op{Coder}(TA)\cong \op{Hom}(TA,A)$ induced by
\be\no
 D\mapsto \pi_1\circ D \;\text{,}
\ee
with inverse (lifting map)
\be\no
d\mapsto \liftd{d}\defineL (\op{id}\tp d \tp \op{id})\circ \Delta_3 \;\text{,}
\ee
where $\Delta_n$ denotes the $n$-fold comultiplication.

Similarly a cohomomorphisms $F\in\op{Cohom}(TA,TA^\prime)$ is a linear map from $TA$ to $TA^\prime$ satisfying
\be\no
\Delta\circ F= (F\tp F)\circ\Delta\;\text{,}
\ee
which implies $\op{Cohom}(TA,TA^\prime)\cong \op{Hom}(TA,A^\prime)$, induced by
\be\no
F\mapsto \pi_1\circ F \;\text{,}
\ee
with inverse (lifting map)
\be\no
f\mapsto \lifte{f}\defineL \sum_{n=0}^\infty f^{\tp n}\circ\Delta_n \;\text{.}
\ee

The Gerstenhaber bracket $[\cdot,\cdot]$ defined by
\be\no
[D_1,D_2]=D_1\circ D_2-(-1)^{D_1 D_2}D_2\circ D_1
\ee
endows $\op{Coder}(TA)$ with the structure of a graded Lie algebra.
Now an $A_\infty$-algebra is defined by a coderivation $M\in \op{Coder}(TA)$ of degree 1 that squares to zero. This in turn makes 
$\op{Coder}(TA)$ a differential graded Lie algebra (DGL) with Hochschild differential $d_h$ defined by
\be\no
d_h=[M,\cdot]\;\text{,}
\ee
and deformations of $M$ are controlled by this DGL. An $A_\infty$-algebra $M$ is denoted as strong or weak, corresponding to whether $\pi_1\circ M\circ i_0$ is zero or non-zero respectively. Let $(A,M)$ and $(A^\prime,M^\prime)$ be $A_\infty$-algebras, then an $A_\infty$-morphism $F\in\op{Morph}(A,A^\prime)$
is a cohomomorphism of degree zero which commutes with the differentials
\be\no
F\circ M= M^\prime \circ F \;\text{.}
\ee
Furthermore $F\in\op{Morph}(A,A^\prime)$ is called an $A_\infty$-quasi-isomorphism if the linear map $\pi_1\circ F\circ i_1$ induces an isomorphism on cohomologies. Similarly it is called an
$A_\infty$-isomorphism if $\pi_1\circ F\circ i_1$ defines an isomorphism. We also distinguish between strong and weak $A_\infty$-morphisms, corresponding to whether  $\pi_1\circ f\circ i_0$ is zero or non-zero respectively.

A Maurer-Cartan element of an $A_\infty$-algebra $(A,M)$ is a degree zero element $a\in A$ that satisfies
\be\no
M(e^a)=0 \;\text{,}
\ee
that is $e^a$ is a constant (no inputs) $A_\infty$-morphism on $A$. The space of all Maurer-Cartan elements is denoted by $\mathcal{MC}(A,M)$.

Furthermore we have the notion of gauge equivalence on the space of Maurer-Cartan elements: Gauge transformations are implemented by a family of $A_\infty$-isomorphisms 
$U_\lambda(t)$, defined by
\be\no
\frac{d}{dt}U_\lambda(t)=[M,\liftd{\lambda}(t)]\circ U_\lambda(t) \mand U_\lambda(0)=\op{id}\;\text{,}
\ee
where $\lambda(t)\in A$ is of degree minus 1 \cite{Kajiura open1}. The moduli space of an $A_\infty$-algebra is defined to be the Maurer-Cartan space modulo gauge transformations
\be\no
\mathcal{M}(A,M)\defineL\mathcal{MC}(A,M)/\sim\;\text{,}
\ee
that is for $a,b\in \mathcal{MC}(A,M)$, $a\sim b$ if there is a gauge transformation $U_\lambda(t)$ such that $e^b=U_\lambda(1)(e^a)$.

A background shift by an element $a\in A$ of degree zero is implemented by the cohomomorphism $E(a)$ defined by
\be\no
E(a)(a_1\tp\dots\tp a_n)=e^a\tp a_1\tp e^a\tp\dots \tp e^a\tp a_n\tp e^a \;\text{.}
\ee
For a given $A_\infty$-algebra $(A,M)$ the background shifted $A_\infty$-algebra is defined by $M[a]=E(-a)\circ M \circ E(a)$, which makes $E(a)$ 
a weak $A_\infty$-isomorphism.

Suppose the module $A$ is endowed with an odd symplectic structure $\omega:A\tp A\to R$. A coderivation $D\in \op{Coder}(TA)$ is called cyclic, if the map 
\be\no
\omega(\pi_1\circ D, \cdot):TA\tp A \to R
\ee
is cyclic symmetric in $TA\tp A$. The space of cyclic coderivations is denoted by $\op{Coder}^{cycl}(TA)$ and is closed under the Gerstenhaber bracket.

$L_\infty$-algebras are constructed in a similar way, where instead of the tensor algebra $TA$ one considers the symmetric algebra $SA$. The coalgebra structure on $SA$ is given by
\be\no
\Delta (c_1,\cdots,c_n)=\sum_{i=0}^n\sideset{}{^\prime}\sum_\sigma (c_{\sigma_1}\wedge\cdots\wedge c_{\sigma_i})\tp(c_{\sigma_{i+1}}\wedge\cdots\wedge c_{\sigma_n})\;\text{,}
\ee
where $\sum_\sigma^\prime$ indicates the sum over all permutations $\sigma\in S_n$ constraint to $\sigma_1<\cdots<\sigma_i$ and $\sigma_{i+1}<\cdots<\sigma_n$ (unshuffels).
In table \ref{tab:AinftyLinfty} we summarize the definitions from above, together with the corresponding counterparts in the $L_\infty$ context.

\begin{center}
\begin{table}\caption{$A_\infty$- and $L_\infty$-algebras in summary}\label{tab:AinftyLinfty}\vspace{4pt}
\begin{tabular}{c||c|c}
				& $A_\infty$ & $L_\infty$ \\[1pt]\hline\hline
algebra  	& $M^2=0$, $|M|=1$ & $L^2=0$, $|L|=1$ \\[1pt]\hline
morphism & $F\circ M=M^\prime \circ F$ & $F\circ L = L^\prime \circ F$\\[1pt]\hline
lift (coder)  & $\liftd{d}=(\op{id}\tp d \tp \op{id})\circ \Delta_3$ & $\liftd{d}=(d\w \op{id})\circ\Delta$ \\[1pt]\hline
lift (cohom)  & $\lifte{f}=\sum_{n=0}^\infty f^{\tp n}\circ \Delta_n$ & $\lifte{f}=\sum_{n=0}^\infty \inv{n!} f^{\w n}\circ \Delta_n$ \\[1pt]\hline
\multirow{3}{*}{background shift}   & $E(a)(a_1\tp\dots \tp a_n)=$ & $E(c)(c_1\w \dots \w c_n)=$ \\
						     & $ e^a\tp a_1\tp e^a \dots e^a\tp a_n\tp e^a$ & $ e^c\w c_1\w\dots \w c_n$ \\[1pt]\cline{2-3}
						     & $M[a]=E(-a)\circ M\circ E(a)$ & $L[c]=E(-c)\circ L\circ E(c)$ \\[1pt]\hline
Maurer-Cartan element & $M(e^a)=0$, $|a|=0$ & $L(e^c)=0$, $|c|=0$ \\[1pt]\hline
gauge transformation & $\frac{d}{dt}U_\lambda(t)=[M,\liftd{\lambda}(t)]\circ U_\lambda(t)$ & $\frac{d}{dt}U_\lambda(t)=[L,\liftd{\lambda}(t)]\circ U_\lambda(t)$ \\[1pt]\hline
cyclicity & $\omega(\pi_1\circ D, \cdot)$ cyclic sym. & $\omega(\pi_1\circ D, \cdot)$ full sym.
\end{tabular}
\end{table}
\end{center}

\subsection{$IBL_\infty$-algebras}\label{inftya}
Homotopy involutive Lie bialgebras ($IBL_\infty$-albegras) as presented in \cite{Cieliebak ibl}, are constructed similarly to $L_\infty$-algebras. The definition includes an external parameter $\hbar$ and makes use of higher order coderivations \cite{Akman coder, Bering coder, Markl loop}. In the previous section we saw that (first oder) coderivations on $SA$ are in one-to-one correspondence with homomorphisms from $SA$ to $A$, where the correspondence is established by the lifting map and the projection $\pi_1$. That is, a first order coderivations is the lift of a homomorphism with an arbitrary number of inputs and one output.  Higher order coderivations are then introduced by allowing for several outputs of the homomorphism:
The space $\op{Coder}^n(SA)$ of coderivations of order $n$ is isomorphic to $\op{Hom}(SA,\Sigma^nA)$, where $\Sigma^nA\defineL \bigoplus_{i=1}^nA^{\w i}$. The isomorphism is given by
\be\no
\begin{array}{rcl}
\op{Hom}(SA,\Sigma^nA)&\to& \op{Coder}^n(SA)\\[3pt]
d & \mapsto & \liftd{d}=(d\w\op{id})\circ\Delta\;\text{,}
\end{array}
\ee 
with inverse
\be\no
\begin{array}{rcl}
\op{Coder}^n(SA) &\to& \op{Hom}(SA,\Sigma^nA)\\[3pt]
D &\mapsto&\left\{ \bnm \pi_1\circ D \\+ \bracketii{\pi_2\circ D -(\pi_1\circ D\w\pi_1)\circ\Delta} \\\vdots\\ +\bracketii{\pi_n\circ D -\sum_{i+j=n-1}(\pi_i\circ D\w \pi_{j+1})\circ\Delta} \enm \right. \;\text{.}
\end{array}
\ee 

The graded commutator 
\be\no
[D_1,D_2]=D_1\circ D_2-(-1)^{D_1D_2}D_2\circ D_1 \;\text{,}
\ee
where  $D_1$, $D_2$ are arbitrary higher order coderivations, satisfies the property
\be\label{eq:coderprop}
[\op{Coder}^i(SA),\op{Coder}^j(SA)]=\op{Coder}^{i+j-1} \;\text{.}
\ee

Consider now the space 
\be\no
\coder(SA,\hbar)\defineL \bigoplus_{n=1}^\infty \hbar^{n-1} \op{Coder}^n(SA) \;\text{.}
\ee
Equation (\ref{eq:coderprop}) implies that the graded commutator raises $\coder(SA,\hbar)$ to a graded Lie algebra.
From the definition of higher order coderivations, we obtain the isomorphism $\coder(SA,\hbar)\cong \bigoplus_{n=1}^\infty \hbar^{n-1}\op{Hom}(SA,\Sigma^nA)$.
For an element $\frak{d}\in \bigoplus_{n=1}^\infty \hbar^{n-1}\op{Hom}(SA,\Sigma^nA)$ we define associated maps $d^{n,g}\in \op{Hom}(SA,A^n)$ by
\be\no
\frak{d}=\sum_{n=1}^\infty\sum_{g=0}^\infty \hbar^{n+g-1}\,d^{n,g} \;\text{,}
\ee
that is we expand $\frak{d}$ in the number of outputs.

The definition of $IBL_\infty$-algebras, $IBL_\infty$-morphisms, etc. resembles that of $L_\infty$-algebras, except that we substitute $\op{Hom}(SA,A)$ by 
$\bigoplus_{n=1}^\infty \hbar^{n-1}\op{Hom}(SA,\Sigma^nA)$:
An $IBL_\infty$-algebra is defined by an element $\frak{L}\in\coder(SA,\hbar)$ of degree one that squares to zero. A cohomomorphism $\frak{F}\in\cohom(SA,SA^\prime,\hbar)$ is determined by a map 
$\frak{f}\in\bigoplus_{n=1}^\infty \hbar^{n-1}\op{Hom}(SA,\Sigma^nA)$ via the lifting map
\be\no
\frak{F}=e^{\frak{f}}=\sum_{n=0}^\infty\inv{n!}\frak{f}^{\w n} \;\text{.}
\ee
Let $(A,\frak{L})$ and $(A^\prime,\frak{L}^\prime)$ be $IBL_\infty$-algebras. An $IBL_\infty$-morphism $\frak{F}\in\morph(A,A^\prime)$ from $(A,\frak{L})$ to $(A^\prime,\frak{L}^\prime)$ is a cohomomorphism 
of degree zero, that commutes with the differentials
\be\no
\frak{F}\circ\frak{L}=\frak{L}^\prime\circ\frak{F} \;\text{.}
\ee
Similarly a Maurer-Cartan element of an $IBL_\infty$-algebra $(A,\frak{L})$ is an element $\frak{c}\in\bigoplus_{n=1}^\infty \hbar^{n-1}\Sigma^nA$  of degree zero, satisfying
\be\no
\frak{L}(e^\frak{c})=0 \;\text{.}
\ee
The space of Maurer-Cartan elements of an $IBL_\infty$-algebra is denoted by $\mathcal{MC}(A,\frak{L})$. In analogy to the $L_\infty$ case we define gauge transformations  by a family of 
$IBL_\infty$-isomorphisms $\frak{U}_{\frak{\lambda}}(t)$ determined by
\be\no
\frac{d}{dt}\frak{U}_{\frak{\lambda}}(t)=[\frak{L},\liftd{\frak{\lambda}}(t)]\circ\frak{U}_{\frak{\lambda}}(t) \mand \frak{U}_{\frak{\lambda}}(0)=\op{id} \;\text{,}
\ee
where $\lambda(t)\in\bigoplus_{n=1}^\infty \hbar^{n-1}\Sigma^nA$ is of degree minus one.
Finally the moduli space of an $IBL_\infty$-algebra is the space of Maurer-Cartan elements modulo gauge transformations, that is $\mathcal{M}(A,\frak{L})=\mathcal{MC}(A,\frak{L})/\sim$.
Obviously one recovers the $L_\infty$ structures in the limit $\hbar\to0$.


\end{document}